\documentclass[a4paper,UKenglish,cleveref, autoref, thm-restate]{lipics-v2021}

\pdfoutput=1
\hideLIPIcs
\nolinenumbers

\usepackage{algorithm}
\usepackage{amsmath}
\usepackage{mathtools}
\usepackage{ stmaryrd }
\usepackage[noend]{algpseudocode}
\usepackage{varwidth}
\usepackage{tikz}
\usepackage{subcaption}
\usetikzlibrary{automata, graphs,positioning,chains,arrows,snakes,decorations.pathmorphing}

\newcommand{\sem}[1]{{\llbracket#1\rrbracket}}
\newcommand{\cA}{\mathcal{A}}
\newcommand{\cM}{\mathcal{M}}
\newcommand{\bbT}{\mathbb{T}}
\newcommand{\bbD}{\mathbb{D}}
\newcommand{\bbN}{\mathbb{N}}
\newcommand{\frakS}{\mathfrak{S}}

\newcommand{\str}{{\sf str}}

\title{Dynamic direct (ranked) access of MSO query evaluation over SLP-compressed strings} %

\titlerunning{Dynamic direct (ranked) access of MSO query eval. over SLP-compressed strings} %

\author{Martín Muñoz}
{Univ. Artois, CNRS, UMR 8188, Centre de Recherche en Informatique de Lens (CRIL), F-62300, France}
{munoz@cril.fr}
{https://orcid.org/0009-0003-3294-6159}
{}

\authorrunning{Martín Muñoz} %

\Copyright{Martín Muñoz} %

\ccsdesc[500]{Theory of computation~Logic and databases}
\ccsdesc[300]{Theory of computation~Data structures and algorithms for data management}
\ccsdesc[300]{Theory of computation~Automata extensions}

\keywords{Direct Access, MSO, incremental evaluation, text compression, algorithms on compressed text} %

\begin{document}

\maketitle

\begin{abstract}
	We present an algorithm that, given an index $t$, produces the $t$-th (lexicographically ordered) answer of an MSO query over a string.
	The algorithm requires linear-time preprocessing, and builds a data structure that answers each of these calls in logarithmic time. 
	We then show how to extend this algorithm for a string that is compressed by a straight-line program (SLP), also with linear-time preprocessing in the (compressed encoding of the) string, and maintaining direct access in logtime of the original string.
	Lastly, we extend the algorithm by allowing complex edits on the SLP after the direct-access data structure has been processsed, which are translated into the data structure in logtime. We do this by adapting a document editing framework introduced by Schmid and Schweikardt (PODS 2022).
	This work improves on a recent result of dynamic direct access of MSO queries over strings (Bourhis et. al., ICDT 2025) by a log-factor on the access procedure, and by extending the results to SLPs.
\end{abstract}

\section{Introduction}\label{sec:intro}

Evaluating queries built from monadic-second order logic (MSO) formulas is a well-studied task in database theory and algorithmics for formal languages. 
A reason for this is that it serves as a pivot to other formalisms.
In particular, owing to the Büchi-Elgot-Trakhtenbrot theorem~\cite{buchi}, it is known that MSO over words is equivalent to finite state automata and regular expressions. This equivalence also extends to trees, as MSO queries over trees are equivalent to tree regular languages and visibly pushdown automata. 
The reach of the task extends further, since the regex formulation means MSO evaluation has use in string processing; the tree formulation means it has use in XML processing; and very relevantly, MSO can be used to handle evaluation for regular spanners~\cite{spannersenum1, rankedenum, nestedenum}.

The base evaluation task is: given a MSO formula $\varphi$, a data string $w$, does $w\models \varphi$ hold? (understanding $w$ as an ordered monadic structure in the typical sense). 
A natural extension of this task is to instead query via a formula with open variables $\varphi(x_1,\ldots,x_k)$, and to compute the assignments $\mu$ such that $w\models\varphi(\mu(x_1),\ldots,\mu(x_k))$.

Naturally, by this extension, the set of answers can become much larger than the input instance. 
To overcome this, a great deal of works have studied different ways of extracting the answers without producing them all at once. 
In MSO evaluation, we find works that explore the task of enumerating these answers one by one~\cite{bagan, courcelle, antoineicalp}; we also find {\em ranked enumeration}, where the answers have to be produced in order~\cite{rankedenum}; and {\em direct access}, where the set of answers is accessed on an arbitrary index~\cite{phdbagan, phdkazana}.

We identify two different tasks that share a name in the literature. First is {\em random access}: for the set of answers $S$, the goal is to construct an index between $1$ and $|S|$ such that each produces a different answer from the set. Second is {\em ranked access}: the task is the same but there is an order on the answers, and the index follows this order. The naming in the literature shows discrepancies and both tasks have also been named {\em direct access}~\cite{nofar, count1, count2, billie}. 

We highlight the fact that, as noted in~\cite{nofar}, direct random access can simulate uniform sampling and enumeration; whereas direct ranked access encompasses quantile queries (e.g., the median of the answers), ranked enumeration and top-k answering.

On top of the evaluation task, some works venture into the realm of {\em incremental evaluation}. In other words, we allow changes in the data after it has been processed for the query. 
For MSO queries, this has been abundantly explored in enumeration~\cite{treerelabels, niewerthtrees, compressedtreesenum, slpedits, slpenum2}. However, as far as we know, only one work~\cite{diraccmso} has studied direct access for MSO on an editable string.

The present work studies an algorithm for efficient direct ranked access on MSO queries. Further, we also consider a more complex setting: the string may have been compressed as a grammar with a singleton language~\cite{rubin}. Encodings of this form are typically called straight-line programs (SLP)~\cite{slp}; they enjoy good compression bounds and have enticing algorithmic properties~\cite{billie, navarrosurvey1, navarrosurvey2, urbina, slpsurvey}.
Evaluation for MSO queries on SLP-compressed strings (and trees) has been already considered for enumeration~\cite{slpenum1, compressedtreesenum, slpedits, slpenum2}, yet we believe ours is the first work that approaches it from the angle of direct access.

\subparagraph{Background}
This work is an improvement of a result in~\cite{diraccmso}: an algorithm for dynamic direct access for MSO queries on strings. Their result achieves logsquare-access with log-updates on the size of the input string. 
In this work, we reduce the logsquare-access to a single log factor, maintaining the log-updates.
Further, we include an extension to SLP-compressed strings which ends up appearing naturally due to the construction of the data structure.
We borrow various concepts and insights from this work.

\subparagraph{Paper structure}
On Section~\ref{sec:prelim}, we introduce some known concepts and easy extensions.
On Section~\ref{sec:words}, we present one of the main results of the paper, 
which is the direct access algorithm for MSO queries on strings.
On Section~\ref{sec:slps}, we extend this result to SLPs, 
and on Section~\ref{sec:upd}, we explain how to handle edits.
We give some closing words on Section~\ref{sec:concl}.

Due to space constraints, we defer (at least one) proof to the appendix. 
Statements are given without proof when we believe they follow straightforwardly.

\section{Preliminaries}\label{sec:prelim}

Let $\bbN = \{0,1,\ldots\}$. Given sets $A, B$ we use the notation $B^{A}$ for the set of functions $f:A\to B$. In particular, for a set $A$, we write $\bbN^{A\times A}$ to denote the matrices of natural numbers indexed by pairs $(a_1, a_2)\in A\times A$. We also use the notation $I_A$ for the $A\times A$ identity matrix; i.e., $I_A(a_1, a_2) = 1$ if $a_1 = a_2$ and $I_A(a_1, a_2) = 0$ if $a_1 \neq a_2$.

We use $[1,n] = \{1,\ldots,n\}$. We use the symbol $<$ to denote the standard order over $\bbN$. If a set $A$ is totally ordered by some order $\preccurlyeq$, we will use the word {\em slice} for a subset $B\subseteq A$ that can be described by two elements $a_{l},a_r\in A$ such that $B = \{a\mid a_l \preccurlyeq a \preccurlyeq a_r\}$. 

We write $\omega < 3$ for an exponent for binary matrix multiplication~\cite{Strassen1969}.

\subparagraph{Strings and mappings}
Given a finite alphabet $\Sigma$, a {\em string over $\Sigma$} (or just a {\em string}) is a sequence $w = a_1\ldots a_n\in\Sigma^*$, where $\Sigma^*$ denotes the set of finite sequences of elements in $\Sigma$. For such a $w$, we define its {\em length} as $n$, which we write by $|w|$. The string of length zero is denoted by $\varepsilon$, and we write $\Sigma^+$ for $\Sigma^* \setminus \{\varepsilon\}$.
For strings $w_1$ and $w_2$ we denote their concatenation by $w_1\cdot w_2$.
For a string $w = a_1\ldots a_n$,  for any $1 \leq i \leq j \leq n$ we use the notation $w[i,j] = a_i \ldots a_j$, and $w[i] = w[i,i] = a_i$.

Fix an $n$. Given a finite set of variables $X = \{x_1,\ldots,x_k\}$, a {\em mapping} of $X$ onto $w$ is a function $\mu:X\to[1,n]$. 
Alternatively, we will define mappings as sets of the form $\{x_1\mapsto s_1,\ldots,x_k\mapsto s_k\}$, for $s_1\ldots,s_k\in[1,n]$. A {\em partial mapping} is any subset of a mapping. 
Partial mappings can be also understood as partial functions with domain $X$, of which the total ones are simply mappings.
Fix an order $\prec$ over $X$ such that $x_1 \prec \cdots \prec x_k$. 
We extend this order onto mappings as follows:
Let $\mu_1,\mu_2:X\to[1,n]$ be two mappings, then $\mu_1 \prec \mu_2$ iff $\mu_1(x_1) = \mu_2(x_1), \ldots, \mu_1(x_{i-1}) = \mu_2(x_{i-1})$
and $\mu_1(x_i) \prec \mu_2(x_i)$ for some $i \in [1,k]$.

Given an $x\in X$ and $\alpha\subseteq[1,n]$, we call $x\Mapsto \alpha$ a {\em mapping rule}. 
We say that a (partial) mapping $\mu:X\to[1,n]$ {\em respects} a set  $\tau$ of mapping rules if for every rule $x\Mapsto \alpha \in \tau$ it holds that $\mu(x)\in \alpha$. When $\alpha$ is a singleton set, we may simplify the notation and instead of $x\Mapsto \{s\}$, write it as $x\mapsto s$.
For a fixed $X$,
a mapping rule is {\em functional} if it mentions every $x\in X$.
Two sets $\tau_1, \tau_2$ of mapping rules are {\em equivalent} if every mapping respects $\tau_1$ iff it respects $\tau_2$.
Observe that for every set of mapping rules, there is an equivalent functional~one.

\subparagraph{Vset Automata}
Given a fixed set of variables $X$, a vset automata is a tuple $\cA = (\Sigma, Q, \Delta, I, F)$ where $\Sigma$ is a finite alphabet, $Q$ is a set of states $\Delta \subseteq Q\times (\Sigma \times 2^{X})\times Q$ is the set of transitions, and $I, F\subseteq Q$ are the sets of initial and final states, respectively. 
A {\em run $\rho$ of $\cA$ over} some string $w = a_1\ldots a_n$ is a sequence 
\[
\rho = q_0\xrightarrow{a_1\, ,\,S_1} q_1\xrightarrow{a_2\, ,\, S_2}\cdots \xrightarrow{a_n\, ,\, S_n} q_n 
\]
where $q_0\in I$ and $(q_{i-1},(a_i, S_i),q_i)\in \Delta$ for each $i\in[1,n]$. We say that $\rho$ is {\em valid} if $S_1\cup\cdots \cup S_n = X$, and for each $i, j\in[1,n]$ it holds that $S_i\cap S_j = \emptyset$. In other words, $S_1,\ldots,S_n$ forms a partition of $X$.
For a valid $\rho$, we extract a mapping $\mu_{\rho}$ given by $\mu_{\rho}(x) = i$ where $S_i$ is the (only) set that contains $x$.
We also say that $\rho$ is {\em accepting} if $q_n\in F$. The semantics of vset automata is given by a function
\[
\sem{\cA}(w) := \{\mu_{\rho}\mid \text{$\rho$ is a valid and accepting run of $\cA$ over $w$} \}.
\]

Given $w = a_1\ldots a_n$ and $l,r\in[1,n]$ where $l \leq r$ we say that $\rho'$ is an {\em $(l,r)$-partial run over $w$} if there is a run $\rho$ over $w$ like above, and $\rho' = q_{l-1}\xrightarrow{a_{l}\, ,\,S_{l}} q_{l}\xrightarrow{a_{l+1}\, ,\, S_{l+1}}\cdots \xrightarrow{a_r\, ,\, S_r} q_r $. From an $(l,r)$-partial run $\rho'$, we extract the partial mapping $\mu_{\rho'} := \{x\mapsto i\mid \rho'\text{ has an $S_i\ni x$}\}$.

As in~\cite{diraccmso}, we deal with two restrictions on vset automata: we say that $\cA$ is {\em functional} if every accepting run of $\cA$ is necessarily valid; and we say that $\cA$ is {\em unambiguous} if it holds that for any two runs $\rho_1$ and $\rho_2$ that satisfy $\mu_{\rho_1} = \mu_{\rho_2}$, then $\rho_1 = \rho_2$. Note that for every vset automata there is an equivalent functional unambiguous one~\cite{spans, incompl}. In what follows, we focus only on vset automata that satisfy both restrictions.

We define $\sem{\cA}(w)\langle\tau\rangle := \{\mu\in \sem{\cA}(w)\mid\text{$\mu$ respects $\tau$} \}$ and denote its size by $\sharp\sem{\cA}(w)\langle\tau\rangle$.
As a general rule, the mapping restrictions $\tau$ that we use result in $\sem{\cA}(w)\langle\tau\rangle$ being a slice of $\sem{\cA}(w)$. Namely, the elements in $\sem{\cA}(w)\langle\tau\rangle$ are contiguous inside $\sem{\cA}(w)$ w.r.t. the  order~$\prec$.

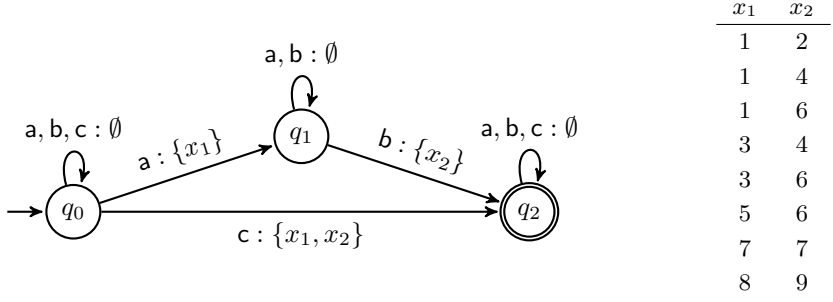
\begin{figure}[t]
	\centering
	
    \begin{subfigure}[r]{0.5\textwidth}            
\begin{tikzpicture}[scale=1.0,->,>=stealth',shorten >=1pt,auto,node distance=2cm,thick,state/.style={circle,draw}, color=black, initial text= {},
		initial distance={5mm}]
		\node[state,initial] (q0) at (0,0) {$q_0$};
		\node[state] (q1) at (3,1) {$q_1$};
		\node[state, accepting] (q2) at (6,0) {$q_2$};
		\draw (q0) to[loop above] node {${\sf a}, {\sf b}, {\sf c} : \emptyset$} (q0);
		\draw (q1) to[loop above] node {${\sf a},{\sf b} : \emptyset$} (q1);
		\draw (q2) to[loop above] node {${\sf a}, {\sf b}, {\sf c}  : \emptyset$} (q2);
		\draw (q0) to node[above, rotate=17] {${\sf a} : \{x_1\}$} (q1);
		\draw (q1) to node[above, rotate=-17] {${\sf b} : \{x_2\}$} (q2);
		\draw (q0) to node[below] {${\sf c} : \{x_1, x_2\}$} (q2);
\end{tikzpicture}	
    \end{subfigure}%
    \hspace{7em}
        \begin{subfigure}{0.2\textwidth}      
\begin{tabular}{ cc}
$x_1$ & $x_2$ \\
\hline
$1$ & $2$  \\ 
$1$ & $4$  \\
$1$ & $6$  \\
$3$ & $4$  \\
$3$ & $6$  \\
$5$ & $6$  \\
$7$ & $7$  \\
$8$ & $9$  
\end{tabular}
    \end{subfigure}%
	\caption{(left) The running example $\cA^{\sf ex}$ illustrated by the automaton, and (right) the list of mappings in $\sem{\cA^{\sf ex}}(w_0)$ for string $w_0= {\sf abababcab}$. }
	\label{fig:re}
\end{figure}

\subparagraph{MSO and vset automata}

In the same line as~\cite{diraccmso}, we present our results as query evaluation for Monadic Second Order on strings: for a MSO formula of the form $\varphi(x_1,\ldots,x_k)$ over strings (with open first-order variables $x_1,\ldots,x_k$), and a string $w\in\Sigma^*$ that is interpreted as a structure, we define $\sem{\varphi}(w) := \{\mu:\{x_1,\ldots,x_k\}\to[1,|w|]\mid \varphi(\mu(x_1),\ldots,\mu(x_k))\}$ as the answer set for the evaluation on $\varphi$ and $w$. 
The results in this work are all stated in terms of vset automata, as for any MSO formula $\varphi(x_1,\ldots,x_k)$ over strings, there exists a vset automaton $\cA$ for $X = \{x_1,\ldots,x_k\}$ that satisfies $\sem{\varphi(x_1,\ldots,x_k)} = \sem{\cA}$.

\subparagraph{Straight-line programs.}
A \emph{context-free grammar} is a tuple $G = (N,\Sigma,R,S_0)$, where $N$ is a non-empty set of non-terminals, $\Sigma$ is finite alphabet, $S_0\in N$ is the start symbol and $R\subseteq N \times(N \cup\Sigma)^{+}$ is the set of rules.
As a convention, the rule $(A,\beta) \in R$ will be written as $A\to \beta$, and we will call $\Sigma$ and $N$ the set of terminal and non-terminal symbols, respectively. 
A context-free grammar $\frakS = (N,\Sigma,R,S_0)$ is a \emph{straight-line program} (SLP) if $R$ is a total function from $N$ to $(N \cup\Sigma)^{+}$ and the directed graph $(N, \{(A,B)\mid (A,\beta) \in R \text{ and } B\text{ appears in } \beta\})$ is acyclic. 
For every $A \in N$, let $R(A)$ be the unique $\beta \in (N \cup\Sigma)^{+}$ such that $(A,\beta) \in R$, and for every $a \in \Sigma$ let $R(a) = a$.
We extend $R$ to a morphism $R^*: (N \cup\Sigma)^*\to \Sigma^*$ recursively such that $R^*(u) = u$ when $u$ is a string over $\Sigma$, and $R^*(\alpha_1 \ldots \alpha_n) = R^*\big(R(\alpha_1) \cdot \ldots \cdot R(\alpha_n)\big)$, where $\alpha_i \in (N \cup \Sigma)$ for every $i \leq n$. 
By our definition of SLP, $R^*(A)$ is in $\Sigma^+$, and uniquely defined for each $A\in N$. 
Then we define the string encoded by $S$ as $\str(S) = R^*(S_0)$.

We say that a context-free grammar $G = (N,\Sigma,R,S_0)$ is in {\em Chomsky normal form} if rules in $R$ are either of the form $A\to BC$ and $A\to a$, for $A, B, C\in N$ or $a\in \Sigma$.

We define the size of an SLP $\frakS = (N,\Sigma,R,S_0)$ as $|\frakS|= \sum_{A \in N} |R(A)|$. 

{\it Example. }
	Let $\frakS = (N, \Sigma, R, S_0)$ be an SLP with $N = \{S_0, A, B, C, D, S_{\sf a}, S_{\sf b}, S_{\sf c}\}$, $\Sigma = \{{\sf a}, {\sf b}, {\sf c}\}$, and $R = \{S_0 \to AB\, ,\, A\to DD\, ,\, B\to CD\, ,\, C\to D S_{\sf c}\, ,\, D\to S_{\sf a}S_{\sf b}\, ,\, S_{\sf a}\to{\sf a}\, , \, S_{\sf b}\to{\sf b}\, ,\, S_{\sf c}\to{\sf c}\}$. We then have that $\str(D) = {\sf ab}$, $\str(C) = {\sf abc}$, $\str(B) = {\sf abcab}$, $\str(A) = {\sf abab}$, and $\str(\frakS) = \str(S_0) = {\sf abababcab}$, namely, the string represented by $\frakS$.
We illustrate the DAG of $R$ in Figure~\ref{fig:slp1} (in the context of the main algorithm for SLPs).

\subparagraph{Machine model}
When dealing with an SLP $\frakS$, it could be that the string $w$ that it represents has length $2^{|\frakS|}$.
In particular, it could happen that merely representing an index in $w$ requires a linear number of bits with respect to $|\frakS|$.

We handle this by defining a function $\kappa$ such that $\kappa(T)$ describes the cost of an arithmetical operation between numbers of value, and result, at most $T$. 
We will deal with numbers that represent indices in $\sem{\cA}(w)$, so the largest value we may ever deal with is $T \leq |w|^{|X|}$. 
We simply assume that every operation we perform has cost $\kappa(|w|^{|X|})$, and we omit this factor from the presented bounds, noting here that it should be regarded as implicit.

Note that, for a string $w$, $\kappa(|w|)$ can be assumed to be constant, and this matches the typical assumptions 
of operating in constant-cost registers of size $\log n$. However, this implies that the bounds contain an extra $|X|$ factor, which, in fact, was properly accounted for in~\cite{diraccmso}. 
We choose to keep the factor implicit nonetheless.

\subsection{Running example}
We introduce an example that we will use now to illustrate the semantics of the model, and later on for explaining the direct access algorithm. We define an unambiguous functional vset automata $\cA^{\sf re}$ as the one depicted in Figure~\ref{fig:re} (left), and we evaluate it over the string $w_0 = {\sf abababcab}$.
We see that the accepting runs of $\cA^{\sf re}$ over $w_0$ that render the first three elements in $\sem{\cA^{\sf re}}(w_0)$ are:
\begin{align*}
\rho_1 &\ =\  q_0\xrightarrow{{\sf a}\, , \, \{x_1\}} 
		q_1\xrightarrow{{\sf b}\, , \, \{x_2\}} 
		q_2\xrightarrow{{\sf a}\, , \, \emptyset} 
		q_2\xrightarrow{\ \ {\sf b}\, , \, \emptyset\ \ } 
		q_2\xrightarrow{{\sf a}\, , \, \emptyset} 
		q_2\xrightarrow{\ \ {\sf b}\, , \, \emptyset\ \ } 
		q_2\xrightarrow{{\sf c}\, , \, \emptyset} 
		q_2\xrightarrow{{\sf a}\, , \, \emptyset} 
		q_2\xrightarrow{{\sf b}\, , \, \emptyset} 
		q_2\, ,\\
\rho_2 &\ =\  q_0\xrightarrow{{\sf a}\, , \, \{x_1\}} 
		q_1\xrightarrow{\ \ {\sf b}\, , \, \emptyset\ \ } 
		q_1\xrightarrow{{\sf a}\, , \, \emptyset} 
		q_1\xrightarrow{{\sf b}\, , \, \{x_2\}} 
		q_2\xrightarrow{{\sf a}\, , \, \emptyset} 
		q_2\xrightarrow{\ \ {\sf b}\, , \, \emptyset\ \ } 
		q_2\xrightarrow{{\sf c}\, , \, \emptyset} 
		q_2\xrightarrow{{\sf a}\, , \, \emptyset} 
		q_2\xrightarrow{{\sf b}\, , \, \emptyset} 
		q_2\, ,\\
\rho_3 &\ =\  q_0\xrightarrow{{\sf a}\, , \, \{x_1\}} 
		q_1\xrightarrow{\ \ {\sf b}\, , \, \emptyset\ \ } 
		q_1\xrightarrow{{\sf a}\, , \, \emptyset} 
		q_1\xrightarrow{\ \ {\sf b}\, , \, \emptyset\ \ } 
		q_1\xrightarrow{{\sf a}\, , \, \emptyset} 
		q_1\xrightarrow{{\sf b}\, , \, \{x_2\}} 
		q_2\xrightarrow{{\sf c}\, , \, \emptyset} 
		q_2\xrightarrow{{\sf a}\, , \, \emptyset} 
		q_2\xrightarrow{{\sf b}\, , \, \emptyset} 
		q_2\, ,
\end{align*}
and so, $\mu_{\rho_1} = \{x_1\mapsto 1, x_2\mapsto 2\}$, $\mu_{\rho_2} = \{x_1\mapsto 1, x_2\mapsto 4\}$ and  $\mu_{\rho_3} = \{x_1\mapsto 1, x_2\mapsto 6\}$. The full set of mappings $\sem{\cA^{\sf re}}(w_0)$ can be seen in Figure~\ref{fig:re} (right).

\section{Direct access for MSO queries on strings}\label{sec:words}

This section will be dedicated to proving the first main result in the paper.

\begin{theorem}\label{theo:words}
	Fix a set of variables $X$. Let $\cA$ be an unambiguous vset automata and let $w\in\Sigma^*$.
	\begin{enumerate}
	\item For a fixed order $\prec$ over $X$, we can have direct (ranked) access on the set $\sem{\cA}(w)$  in $O(|Q|^{\omega}\cdot|X|^2\cdot\log(|w|))$ time
	after $O(|Q|^{\omega}\cdot|X|\cdot|w|)$ preprocessing.
	\item If $\prec$ is given as input in the access phase, the problem requires $O(|Q|^{\omega}\cdot 2^{|X|}\cdot|w|)$ 
	preprocessing time, and the same direct (ranked) access bounds hold. 
	\end{enumerate}
\end{theorem}

\subparagraph{Structure of the solution.} 
First, we will describe a generic procedure to retrieve the $t$-th answer of the set $\sem{\cA}(w)$ that assumes we can compute certain data from $\cA$ and $w$.
After this, we define a collection of matrices that describe the logical information that we need so we can retrieve this data.
Then, we will describe a data structure that stores a succinct version of these matrices.
Lastly, we will describe the actual access phase, which uses a caching strategy to both navigate the search space and
compute the relevant data in the same logtime pass.

Throughout the section, we will focus on proving item (1), which means that the order $\prec$ is globally fixed over $X$. We will assume that $X = \{x_1,\ldots,x_k\}$ and that $x_1 \prec \ldots \prec x_k$. At the end of the section, we will describe how to adapt the solution for item (2).

\subsection{Template for the access phase} \label{sec:template}

Let $t$ be the queried index, and let $\mu^*$ be the $t$-th mapping in the set $\sem{\cA}(w)$. The access phase consists in computing the values $\mu^*(x_1),\ldots,\mu^*(x_k)$ in order, one by one.

\begin{algorithm}[t]
	\caption{Simplified algorithm to retrieve the $t$-th element in $\sem{\cA}(w)$.}\label{alg:acc1}	
	\smallskip
	\begin{algorithmic}[1]
		\Procedure{{\sc Access}}{$\cA, w, t$} \label{alg1}
		\State {\bf declare} $s_1,\ldots,s_k$
		\For{$i \in [1,k]$}
		\State $(l,r) \gets (1, n)$
		\While{$l < r$}
		\State $m \gets {\sf mid}(l, r)$
		\If{$t \leq \sharp\sem{\cA}(w)\langle x_1\mapsto  s_1\, ,\, \ldots \, , \, x_{i-1} \mapsto s_{i-1}\, ,\, x_i \Mapsto [1, m]\rangle$}
		\State $(l, r) \gets (l, m)$
		\Else
		\State $(l, r) \gets (m+1, r)$
		\EndIf
		\EndWhile
		\State $s_i \gets l$
		\State $t \gets t - \sharp\sem{\cA}(w)\langle x_1\mapsto  s_1\, ,\, \ldots \, , \, x_{i-1} \mapsto s_{i-1}\, ,\, x_i \Mapsto [1, s_i - 1]\rangle$
		\EndFor
		\State {\bf return} $\{x_1\mapsto s_1\, ,\, \ldots \, , \, x_k \mapsto s_k\}$
		\EndProcedure
	\end{algorithmic}
\end{algorithm}

The access phase is described in Algorithm~\ref{alg:acc1}. 
The base idea is to find $\mu^*(x_1)$ with a binary search; then find $\mu^*(x_2)$ with a binary search that assumes the value $\mu^*(x_1)$ is fixed; then $\mu^*(x_3)$ with a binary search that assumes $\mu^*(x_1)$ and $\mu^*(x_2)$ are fixed; and so on.
The first binary search looks for $\mu^*(x_1)$ by characterizing it as a value $s_1$ that satisfies
\[
\sharp\sem{\cA}(w) \langle x_1\Mapsto[1,s_1-1]\rangle < t \leq \sharp\sem{\cA}(w) \langle x_1\Mapsto[1,s_1]\rangle.
\]
In the algorithm, we assume that we have access to a function ${\sf mid}:[1, n]\times[1,n]\to[1,n]$ that assigns a certain middle point $m := {\sf mid}(l, r)$ to every range $(l,r)$; for now, we only assume that this function is fixed at the beginning of the access phase, and that $l \leq m < r$.

The binary search in the algorithm iterates a range $(l, r)$ that starts as $(1, n)$ and is always guaranteed to contain the target value $s$.
As an invariant, we maintain the inequality $\sharp\sem{\cA}(w) \langle x_1\Mapsto[1,l]\rangle < t \leq \sharp\sem{\cA}(w)\langle x_1\Mapsto[1,r]\rangle$.
Each iteration decides whether to reduce $(l, r)$ into its first half $(l,m)$ or into its second half $(m+1,r)$, where $m = {\sf mid}(l, r)$, by counting how many answers in $\sem{\cA}(w)$ map $x_1$ to the range $(1, m)$. 
When the range reaches size 1 (i.e., $l = r$), we have only one choice for $s_1$ so we assign it to $\mu^*(x_1)$. 
Then, we subtract from $t$ the number of answers that map $x_1$ to a value lower than $s_1$ and continue to look for $\mu^*(x_2)$. Subtracting this value from $t$ lets us focus on the slice of $\sem{\cA}(w)$ that assumes that $x_1$ is mapped to $s_1$. The rest of $\mu^*$ is computed the same way.

In our running example, we can see the full list of mappings in $\sem{\cA^{\sf ex}}(w_0)$ in Figure~\ref{fig:re} (right). If the access index was $t = 5$, then we would find $s_1 = 3$ since $\sharp\sem{\cA^{\sf ex}}(w_0) \langle x_1\Mapsto[1,3-1]\rangle = 3 < 5$ and $\sharp\sem{\cA^{\sf ex}}(w_0) \langle x_1\Mapsto[1,3]\rangle = 5 \leq 5$. 
Then, before the next iteration, we need to subtract the value $\sharp\sem{\cA^{\sf ex}}(w_0) \langle x_1\Mapsto[1,3-1]\rangle = 3$ from $t$, and do the binary search with $t = 2$. This is because the search space was reduced to only the mappings that start with $x_1\mapsto 3$. In the end, we find $s_2 = 6$ given that $\sharp\sem{\cA^{\sf ex}}(w_0) \langle x_1\mapsto 3\,,\, x_2\Mapsto[1,6-1]\rangle = 1 < 2$ and $\sharp\sem{\cA^{\sf ex}}(w_0) \langle x_1\mapsto 3\, , \,x_2\Mapsto[1,6]\rangle = 2 \leq 2$ and so $\mu^* = \{x_1\mapsto 3, x_2\mapsto 6\}$,

\begin{lemma}\label{lem:alg1}
	 Algorithm~\ref{alg:acc1} finds the $t$-th answer in $\sem{\cA}(w)$.
\end{lemma}

The generic algorithm gives us the following insight: the relevant slices of $\sem{\cA}(w)$ are those defined by a fixed assignment for $x_1,\ldots,x_{i-1}$, and a rule of mapping of $x_i$ somewhere inside the range $[1,m]$ for some $m$. This will be crucial for our algorithm.

The running time of the procedure depends on the choice of {\sf mid} and the time needed to compute the sizes of these relevant slices. We leave this analysis for later.
\subparagraph{Links to the literature.} This idea of fixing a prefix of the variables to values, and~then computing the size of the corresponding slice of $\sem{\cA}(w)$ tends to be a core subroutine for~direct access~\cite{count1, count2}. 
We think that, conceptually, the most important refinement of our solution is to include the rule ``map the next variable in the prefix to a range'' to the core subroutine.

\subsection{The $\cM$ matrices}

As a way to bridge the gap between the relevant slices of $\sem{\cA}(w)$ and the data structure that we will present, we introduce a useful collection of matrices in $\bbN^{Q\times Q}$. 

First, we add some new features to mapping rules and their sets. We say that a mapping rule $x\Mapsto \alpha$ is {\em inside} $[l,r]$ if $\alpha\subseteq[l,r]$. Also, we will compose functional sets $\tau_1, \tau_2$ of mapping rules and define the operation $\tau_1\cdot \tau_2  := \{x\Mapsto\alpha\cup\beta\mid x\Mapsto\alpha\in\tau_1\text{ and }x\Mapsto\beta\in\tau_2\}$. Observe that if $\tau_1$ is a set of mapping rules inside $[l, m]$ and $\tau_2$ is a set of mapping rules inside $[m+1, r]$, then $\tau_1\cdot\tau_2$ is a set of mapping rules inside $[l, r]$.

For $1 \leq l \leq r \leq n$, and a set $\tau$ of mapping rules inside $[l,r]$, we define this matrix~in~$\bbN^{Q\times Q}$:
\[
\cM\langle l, r:\tau\rangle(p, q) := |\{\rho\mid\rho\text{ is an $(l, r)$-partial run of $\cA$ over $w$ from $p$ to $q$ \,s.t.\! $\mu_\rho$ respects $\tau$}\}|
\]
The principle behind these matrices is to group the partial mappings that make up the mappings in $\sem{\cA}(w)$ by three criteria: the range inside $[1,n]$ that the partial mapping was obtained from, the positions in $[1,n]$ where the variables in $X$ can appear, and the states that delimit their respective partial runs.

In our running example, we highlighted runs $\rho_1,\rho_2,\rho_3$ which render the first three mappings in $\sem{\cA^{\sf ex}}(w_0)$. Let us look at one $\cM$ matrix in particular: let $\tau_{x_1=1} := \{x_1\mapsto 1\}$, and let $(l, r) = (1, 5)$. Then we find that $\cM\langle 1, 5:\tau_{x_1=1}\rangle(q_0,q_1) = 1$, since there is one $(1,5)$-partial run over $w_0$ from $q_0$ to $q_1$ (it can be obtained from $\rho_3$); also, see that $\cM\langle 1, 5:\tau_1\rangle(q_0,q_2) = 2$ since this counts the $(1,5)$-partial runs obtained from $\rho_1$ and $\rho_2$.

The matrices let us summarize a lot of the counting logic in one composition:

\begin{lemma}\label{lem:ms}
	Let $1 \leq l \leq m < r \leq n$, let $\tau_L$ be a functional set of mapping rules inside $[l,m]$ and let $\tau_R$ be a functional set of mapping rules inside $[m+1,r]$. Then:
	\[
		\cM\langle l, r:\tau_L \cdot \tau_R\rangle = \cM\langle l, m:\tau_L\rangle \cdot \cM\langle m+1, r:\tau_R\rangle
	\]
\end{lemma}
Lastly, we observe that the equality 
$\sharp\sem{\cA}(w)\langle\tau\rangle = \sum_{p\in I, q\in F}\cM\langle 1, n : \tau\rangle(p,q)$
 holds for any $\tau$ given that $\cA$ is functional and unambiguous.

\subsection{Preprocessing phase}

\begin{algorithm}[t]
	\caption{Preprocessing phase for direct access on $\sem{\cA}(w)$.}\label{alg:prep}	
	\smallskip
	\begin{algorithmic}[1]
	\hspace{-2.5em}
	\begin{varwidth}[t]{0.5\textwidth}
		\Procedure{{\sc Preprocessing}}{$\cA, w$} 
		\State {\bf declare} $\bbT_1,\ldots,\bbT_k$
		\For{$i \in [0,k]$}
		\State $\textsc{Build}(\bbT_i, 1, n)$
		\EndFor
		\EndProcedure
	\end{varwidth}
	\hspace{2.5em}
	\begin{varwidth}[t]{0.5\textwidth}
		\Procedure{{\sc Build}}{$\bbT_i, l, r$} 
		\State {\sf vars at} $\bbT_i\langle l,r\rangle \gets \{x_{i+1},\ldots,x_k\}$
		\If{$l = r$}
		\State $a\gets w[l]$
		\State $\bbT_i\langle l,r\rangle::M \gets \Delta^{\{x_{i+1}\, ,\,\ldots\, ,\,x_k\}}_a$
		\State {\bf return} $M$
		\Else 
		\State $m \gets {\sf mid}(l, r)$
		\State $M_1\gets \textsc{Build}(\bbT_i, l, m)$
		\State $M_2\gets \textsc{Build}(\bbT_i, m+1, r)$
		\State $\bbT_i\langle l, r\rangle :: M \gets M_1 \cdot M_2$
		\State {\bf return} $M$
		\EndIf
		\EndProcedure
	\end{varwidth}
	\end{algorithmic}
\end{algorithm}

Now we introduce the data structure used in the algorithm, whose main idea is to strategically make some of the $\cM$ matrices explicit. 
Let us now identify these matrices, and show why they are enough to compute the sizes of all the slices of $\sem{\cA}(w)$ 
that are relevant in the access phase.

We define a collection $\{\bbT_Y\}_{Y\subseteq X}$ of binary trees, each of which has $2\cdot n - 1$ nodes, and has the following recursive structure (it is the same across all trees): 
For a $\bbT \in \{\bbT_Y\}_{Y\subseteq X}$,
\begin{itemize}
\item every node in $\bbT$ is indexed by a pair $\langle l, r \rangle$ such that $1\leq l \leq r \leq n$;
\item the root of $\bbT$ is indexed $\bbT \langle 1,n \rangle$;
\item each node $\bbT\langle l, r\rangle$ with $l < r$ has a left child $\bbT\langle l, m\rangle$ and a right child $\bbT\langle m+1, r\rangle$, where $m := {\sf mid}(l, r)$;
\item and when $l = r$, the node $\bbT\langle l,r\rangle$ is a leaf.
\end{itemize}
Again, we assume that there is a fixed function ${\sf mid}$ that assigns a middle point $m$ to $(l, r)$ such that $l \leq m < r$. We call $(l,r)$ a {\em valid index} if it can be deduced by realizing the recursive process above with this ${\sf mid}$.
The only requirement for ${\sf mid}$ is to be {\em strongly balanced}:
after computing a tree $\bbT$ with this choice of ${\sf mid}$, for every node $\bbT\langle l, r\rangle$, it must hold that the height $h_1$ of $\bbT\langle l, m\rangle$, and the height $h_2$ of $\bbT\langle m+1,r\rangle$ satisfy that $h_1 - h_2 \in \{-1,0,1\}$.
We choose to describe the algorithm for a generic function ${\sf mid}$ to be able to reutilize the logic in the  SLP case (Section~\ref{sec:slps}). For strings, one can simply define ${\sf mid}(l, r) = \lfloor\frac{l+r}{2}\rfloor$.

The information that is stored in a $\bbT_{Y}$ in particular can be described by restrictions on the transition set of $\cA$. For $Y\subseteq X$ and $a\in\Sigma$, define $\Delta_a^{Y} := \{(p,(a, S),q)\in\Delta\mid S\subseteq Y\text{ for some $p,q\in Q$}\}$. 
We also overload this notation so that it represents the adjacency matrix of this set. 
In Figure~\ref{fig:tables}, we have written out these matrices for the automaton $\cA^{\sf ex}$, for the $Y$ values  $\{x_1,x_2\}$, $\{x_2\}$ and $\emptyset$.

\begin{figure}
\centering

\begin{tabular}{ c|ccc}
$ \Delta_{\sf a}^{\{x_1,x_2\}}$ & $q_0$ & $q_1$ & $q_2$ \\
\hline
$q_0$ & $1$ &  $1$ &  $0$ \\ 
$q_1$ & $0$ &  $1$ &  $0$ \\ 
$q_2$ & $0$ &  $0$ &  $1$ 
\end{tabular}
\hspace{1em}
\begin{tabular}{ c|ccc}
$\Delta_{\sf b}^{\{x_1,x_2\}}$ & $q_0$ & $q_1$ & $q_2$ \\
\hline
$q_0$ & $1$ &  $0$ &  $0$ \\ 
$q_1$ & $0$ &  $1$ &  $1$ \\ 
$q_2$ & $0$ &  $0$ &  $1$ 
\end{tabular}
\hspace{1em}
\begin{tabular}{ c|ccc}
$\Delta_{\sf c}^{\{x_1,x_2\}}$ & $q_0$ & $q_1$ & $q_2$ \\
\hline
$q_0$ & $1$ &  $0$ &  $1$ \\ 
$q_1$ & $0$ &  $0$ &  $0$ \\ 
$q_2$ & $0$ &  $0$ &  $1$ 
\end{tabular}

\vspace{1em}
\begin{tabular}{ c|ccc}
$\ \, \Delta_{\sf a}^{\{x_2\}}\ \, $ & $q_0$ & $q_1$ & $q_2$ \\
\hline
$q_0$ & $1$ &  $\color{blue}{\bf 0}$ &  $0$ \\ 
$q_1$ & $0$ &  $1$ &  $0$ \\ 
$q_2$ & $0$ &  $0$ &  $1$ 
\end{tabular}
\hspace{1em}
\begin{tabular}{ c|ccc}
$\, \ \Delta_{\sf b}^{\{x_2\}}\, \ $ & $q_0$ & $q_1$ & $q_2$ \\
\hline
$q_0$ & $1$ &  $0$ &  $0$ \\ 
$q_1$ & $0$ &  $1$ &  $1$ \\ 
$q_2$ & $0$ &  $0$ &  $1$ 
\end{tabular}
\hspace{1em}
\begin{tabular}{ c|ccc}
$\, \ \Delta_{\sf c}^{\{x_2\}}\, \ $ & $q_0$ & $q_1$ & $q_2$ \\
\hline
$q_0$ & $1$ &  $0$ &  $\color{blue}{\bf 0}$ \\ 
$q_1$ & $0$ &  $0$ &  $0$ \\ 
$q_2$ & $0$ &  $0$ &  $1$ 
\end{tabular}

\newcommand{\asduasgy}{10.5pt}

\vspace{1em}

\hspace{1pt}
\begin{tabular}{ c|ccc}
$\hspace{\asduasgy} \Delta_{\sf a}^{\emptyset}\hspace{\asduasgy} $& $q_0$ & $q_1$ & $q_2$ \\
\hline
$q_0$ & $1$ &  $0$ &  $0$ \\ 
$q_1$ & $0$ &  $1$ &  $0$ \\ 
$q_2$ & $0$ &  $0$ &  $1$ 
\end{tabular}
\hspace{1em}
\begin{tabular}{ c|ccc}
$\hspace{\asduasgy} \Delta_{\sf b}^{\emptyset} \hspace{\asduasgy}  $ & $q_0$ & $q_1$ & $q_2$ \\
\hline
$q_0$ & $1$ &  $0$ &  $0$ \\ 
$q_1$ & $0$ &  $1$ &  $\color{blue}{\bf 0}$ \\ 
$q_2$ & $0$ &  $0$ &  $1$ 
\end{tabular}
\hspace{1em}
\begin{tabular}{ c|ccc}
$ \hspace{\asduasgy}  \Delta_{\sf c}^{\emptyset} \hspace{\asduasgy} $ & $q_0$ & $q_1$ & $q_2$ \\
\hline
$q_0$ & $1$ &  $0$ &  $0$ \\ 
$q_1$ & $0$ &  $0$ &  $0$ \\ 
$q_2$ & $0$ &  $0$ &  $1$ 
\end{tabular}

\caption{The transition matrices for $\cA^{\sf ex}$ after being restricted by different sets $Y\subseteq X$. The bold blue values highlight the only differences of a matrix with the one above}\label{fig:tables}
\end{figure}

Let us now describe the data stored at every node in each $\bbT_Y$. 
Each $\bbT_Y\langle l,r\rangle$, besides its children (which are accessed by computing $m\gets {\sf mid}(l,r)$ via indices $\langle l, m\rangle$ and $\langle m+1,r\rangle$), contains a matrix $M\in \bbN^{Q\times Q}$, which we notate $\bbT_Y\langle l,r\rangle::M$, and a set of variables which is initialized as $Y$ and indicated in the algorithm by ``$\textsf{vars in }\bbT_Y\langle l,r\rangle$''.

The matrix $M$ stored in $\bbT_Y\langle l,r\rangle$ is recursively defined as follows: 
if $l < r$, then let $M_1$ and $M_2$ be matrices in its left and right child, respectively, and assign $M\gets M_1\cdot M_2$;
otherwise, if $l = r$, then $M$ is assigned $\Delta^Y_a$ where $a = w[l]$. We write out two of these trees for the running example in Figure~\ref{fig:prep}.
\newcommand{\sdssdf}{1.0}
\begin{figure}[t]
	\centering
	
	\resizebox{7cm}{!}{
			\begin{tikzpicture}[xscale=0.75,yscale=0.9,>=stealth',shorten >=1pt,auto,node distance=2cm,dashed,state/.style={circle,draw}, color=black, initial text= {},
		initial distance={5mm}]
		
		\setlength\arraycolsep{1pt}
		\renewcommand\arraystretch{0.8}
		
		\node (1-9) at (0,6) { %
\scalebox{\sdssdf}{
$\begin{pmatrix}
	1 & 1 & 8\\
	0 & 0 & 3\\
	0 & 0 & 1\\
\end{pmatrix}$}};

		\node (1-4) at (-4,4) { %
\scalebox{\sdssdf}{
$\begin{pmatrix}
	1 & 2 & 3\\
	0 & 1 & 2\\
	0 & 0 & 1\\
\end{pmatrix}$}};
		\node (5-9) at (4,4) { %
\scalebox{\sdssdf}{
$\begin{pmatrix}
	1 & 1 & 3\\
	0 & 0 & 1\\
	0 & 0 & 1\\
\end{pmatrix}$}};

		\node (1-2) at (-6,2) { %
\scalebox{\sdssdf}{
$\begin{pmatrix}
	1 & 1 & 1\\
	0 & 1 & 1\\
	0 & 0 & 1\\
\end{pmatrix}$}};
		\node (3-4) at (-2,2) { %
\scalebox{\sdssdf}{
$\begin{pmatrix}
	1 & 1 & 1\\
	0 & 1 & 1\\
	0 & 0 & 1\\
\end{pmatrix}$}};
		\node (5-7) at (2,2) { %
\scalebox{\sdssdf}{
$\begin{pmatrix}
	1 & 0 & 2\\
	0 & 0 & 1\\
	0 & 0 & 1\\
\end{pmatrix}$}};
		\node (8-9) at (6,2) { %
\scalebox{\sdssdf}{
$\begin{pmatrix}
	1 & 1 & 1\\
	0 & 1 & 1\\
	0 & 0 & 1\\
\end{pmatrix}$}};

		\node (1-1) at (-7,0) { %
\scalebox{\sdssdf}{
$\begin{pmatrix}
	1 & 1 & 0\\
	0 & 1 & 0\\
	0 & 0 & 1\\
\end{pmatrix}$}};
		\node (2-2) at (-5,0) { %
\scalebox{\sdssdf}{
$\begin{pmatrix}
	1 & 0 & 0\\
	0 & 1 & 1\\
	0 & 0 & 1\\
\end{pmatrix}$}};
		\node (3-3) at (-3,0) { %
\scalebox{\sdssdf}{
$\begin{pmatrix}
	1 & 1 & 0\\
	0 & 1 & 0\\
	0 & 0 & 1\\
\end{pmatrix}$}};
		\node (4-4) at (-1,0) { %
\scalebox{\sdssdf}{
$\begin{pmatrix}
	1 & 0 & 0\\
	0 & 1 & 1\\
	0 & 0 & 1\\
\end{pmatrix}$}};
		\node (5-6) at (1,0) { %
\scalebox{\sdssdf}{
$\begin{pmatrix}
	1 & 1 & 1\\
	0 & 1 & 1\\
	0 & 0 & 1\\
\end{pmatrix}$}};
		\node (7-7) at (3,0) { %
\scalebox{\sdssdf}{
$\begin{pmatrix}
	1 & 0 & 1\\
	0 & 0 & 0\\
	0 & 0 & 1\\
\end{pmatrix}$}};
		\node (8-8) at (5,0) { %
\scalebox{\sdssdf}{
$\begin{pmatrix}
	1 & 1 & 0\\
	0 & 1 & 0\\
	0 & 0 & 1\\
\end{pmatrix}$}};
		\node (9-9) at (7,0) { %
\scalebox{\sdssdf}{
$\begin{pmatrix}
	1 & 0 & 0\\
	0 & 1 & 1\\
	0 & 0 & 1\\
\end{pmatrix}$}};

		\node (5-5) at (0,-2) { %
\scalebox{\sdssdf}{
$\begin{pmatrix}
	1 & 1 & 0\\
	0 & 1 & 0\\
	0 & 0 & 1\\
\end{pmatrix}$}};
		\node (6-6) at (2,-2) { %
\scalebox{\sdssdf}{
$\begin{pmatrix}
	1 & 0 & 0\\
	0 & 1 & 1\\
	0 & 0 & 1\\
\end{pmatrix}$}};

		\node[above = -0.1cm of 1-9] (1-9b)  {\ {\small $\langle1,9\rangle$:\!\!}};
		\node[above = -0.1cm of 1-4] (1-5b)  {\ {\small $\langle1,4\rangle$:\!\!}};
		\node[above = -0.1cm of 5-9] (6-9b)  {\ {\small $\langle5,9\rangle$:\!\!}};
		\node[above = -0.1cm of 1-2] (1-3b)  {\ {\small $\langle1,2\rangle$:\!\!}};
		\node[above = -0.1cm of 3-4] (4-5b)  {\ {\small $\langle3,4\rangle$:\!\!}};
		\node[above = -0.1cm of 5-7] (6-7b)  {\ {\small $\langle5,7\rangle$:\!\!}};
		\node[above = -0.1cm of 8-9] (8-9b)  {\ {\small $\langle8,9\rangle$:\!\!}};
		\node[above = -0.1cm of 1-1] (1-1b)  {\ {\small $\langle1,1\rangle$:\!\!}};
		\node[above = -0.1cm of 2-2] (2-2b)  {\ {\small $\langle2,2\rangle$:\!\!}};
		\node[above = -0.1cm of 3-3] (3-3b)  {\ {\small $\langle3,3\rangle$:\!\!}};
		\node[above = -0.1cm of 4-4] (4-4b)  {\ {\small $\langle4,4\rangle$:\!\!}};
		\node[above = -0.1cm of 5-6] (6-7b)  {\ {\small $\langle5,6\rangle$:\!\!}};
		\node[above = -0.1cm of 7-7] (7-7b)  {\ {\small $\langle7,7\rangle$:\!\!}};
		\node[above = -0.1cm of 8-8] (8-8b)  {\ {\small $\langle8,8\rangle$:\!\!}};
		\node[above = -0.1cm of 9-9] (9-9b)  {\ {\small $\langle9,9\rangle$:\!\!}};
		\node[above = -0.1cm of 5-5] (5-5b)  {\ {\small $\langle5,5\rangle$:\!\!}};
		\node[above = -0.1cm of 6-6] (6-6b)  {\ {\small $\langle6,6\rangle$:\!\!}};
		\node[below = -0.1cm of 1-1] (1-1b)  {$a$};
		\node[below = -0.1cm of 2-2] (2-2b)  {$b$};
		\node[below = -0.1cm of 3-3] (3-3b)  {$a$};
		\node[below = -0.1cm of 4-4] (4-4b)  {$b$};
		\node[below = -0.1cm of 7-7] (7-7b)  {{$c$}};
		\node[below = -0.1cm of 8-8] (8-8b)  {{$a$}};
		\node[below = -0.1cm of 9-9] (9-9b)  {{$b$}};
		\node[below = -0.1cm of 5-5] (5-5b)  {{$a$}};
		\node[below = -0.1cm of 6-6] (6-6b)  {{$b$}};

\draw (1-9) to (1-4);
\draw (1-9) to (5-9);
\draw (1-4) to (1-2);
\draw (1-4) to (3-4);
\draw (5-9) to (5-7);
\draw (5-9) to (8-9);
\draw (1-2) to (1-1);
\draw (1-2) to (2-2);
\draw (3-4) to (3-3);
\draw (3-4) to (4-4);
\draw (5-7) to (5-6);
\draw (5-7) to (7-7);
\draw (8-9) to (8-8);
\draw (8-9) to (9-9);
\draw (5-6) to (5-5);
\draw (5-6) to (6-6);
 
	\end{tikzpicture}	
		}
		\!\!\!\!
	\resizebox{7cm}{!}{
			\begin{tikzpicture}[xscale=0.75,yscale=0.9,>=stealth',shorten >=1pt,auto,node distance=2cm,dashed,state/.style={circle,draw}, color=black, initial text= {},
		initial distance={5mm}]

		\setlength\arraycolsep{1pt}
		\renewcommand\arraystretch{0.8}
			
		\node (1-9) at (0,6) { %
\scalebox{\sdssdf}{
$\begin{pmatrix}
	1 & 0 & 0\\
	0 & 0 & 3\\
	0 & 0 & 1\\
\end{pmatrix}$}};

		\node (1-4) at (-4,4) { %
\scalebox{\sdssdf}{
$\begin{pmatrix}
	1 & 0 & 0\\
	0 & 1 & 2\\
	0 & 0 & 1\\
\end{pmatrix}$}};
		\node (5-9) at (4,4) { %
\scalebox{\sdssdf}{
$\begin{pmatrix}
	1 & 0 & 0\\
	0 & 0 & 1\\
	0 & 0 & 1\\
\end{pmatrix}$}};

		\node (1-2) at (-6,2) { %
\scalebox{\sdssdf}{
$\begin{pmatrix}
	1 & 0 & 0\\
	0 & 1 & 1\\
	0 & 0 & 1\\
\end{pmatrix}$}};
		\node (3-4) at (-2,2) { %
\scalebox{\sdssdf}{
$\begin{pmatrix}
	1 & 0 & 0\\
	0 & 1 & 1\\
	0 & 0 & 1\\
\end{pmatrix}$}};
		\node (5-7) at (2,2) { %
\scalebox{\sdssdf}{
$\begin{pmatrix}
	1 & 0 & 0\\
	0 & 0 & 1\\
	0 & 0 & 1\\
\end{pmatrix}$}};
		\node (8-9) at (6,2) { %
\scalebox{\sdssdf}{
$\begin{pmatrix}
	1 & 0 & 0\\
	0 & 1 & 1\\
	0 & 0 & 1\\
\end{pmatrix}$}};

		\node (1-1) at (-7,0) { %
\scalebox{\sdssdf}{
$\begin{pmatrix}
	1 & 0 & 0\\
	0 & 1 & 0\\
	0 & 0 & 1\\
\end{pmatrix}$}};
		\node (2-2) at (-5,0) { %
\scalebox{\sdssdf}{
$\begin{pmatrix}
	1 & 0 & 0\\
	0 & 1 & 1\\
	0 & 0 & 1\\
\end{pmatrix}$}};
		\node (3-3) at (-3,0) { %
\scalebox{\sdssdf}{
$\begin{pmatrix}
	1 & 0 & 0\\
	0 & 1 & 0\\
	0 & 0 & 1\\
\end{pmatrix}$}};
		\node (4-4) at (-1,0) { %
\scalebox{\sdssdf}{
$\begin{pmatrix}
	1 & 0 & 0\\
	0 & 1 & 1\\
	0 & 0 & 1\\
\end{pmatrix}$}};
		\node (5-6) at (1,0) { %
\scalebox{\sdssdf}{
$\begin{pmatrix}
	1 & 0 & 0\\
	0 & 1 & 1\\
	0 & 0 & 1\\
\end{pmatrix}$}};
		\node (7-7) at (3,0) { %
\scalebox{\sdssdf}{
$\begin{pmatrix}
	1 & 0 & 0\\
	0 & 0 & 0\\
	0 & 0 & 1\\
\end{pmatrix}$}};
		\node (8-8) at (5,0) { %
\scalebox{\sdssdf}{
$\begin{pmatrix}
	1 & 0 & 0\\
	0 & 1 & 0\\
	0 & 0 & 1\\
\end{pmatrix}$}};
		\node (9-9) at (7,0) { %
\scalebox{\sdssdf}{
$\begin{pmatrix}
	1 & 0 & 0\\
	0 & 1 & 1\\
	0 & 0 & 1\\
\end{pmatrix}$}};

		\node (5-5) at (0,-2) { %
\scalebox{\sdssdf}{
$\begin{pmatrix}
	1 & 0 & 0\\
	0 & 1 & 0\\
	0 & 0 & 1\\
\end{pmatrix}$}};
		\node (6-6) at (2,-2) { %
\scalebox{\sdssdf}{
$\begin{pmatrix}
	1 & 0 & 0\\
	0 & 1 & 1\\
	0 & 0 & 1\\
\end{pmatrix}$}};

		\node[above = -0.1cm of 1-9] (1-9b)  {\ {\small $\langle1,9\rangle$:\!\!}};
		\node[above = -0.1cm of 1-4] (1-5b)  {\ {\small $\langle1,4\rangle$:\!\!}};
		\node[above = -0.1cm of 5-9] (6-9b)  {\ {\small $\langle5,9\rangle$:\!\!}};
		\node[above = -0.1cm of 1-2] (1-3b)  {\ {\small $\langle1,2\rangle$:\!\!}};
		\node[above = -0.1cm of 3-4] (4-5b)  {\ {\small $\langle3,4\rangle$:\!\!}};
		\node[above = -0.1cm of 5-7] (6-7b)  {\ {\small $\langle5,7\rangle$:\!\!}};
		\node[above = -0.1cm of 8-9] (8-9b)  {\ {\small $\langle8,9\rangle$:\!\!}};
		\node[above = -0.1cm of 1-1] (1-1b)  {\ {\small $\langle1,1\rangle$:\!\!}};
		\node[above = -0.1cm of 2-2] (2-2b)  {\ {\small $\langle2,2\rangle$:\!\!}};
		\node[above = -0.1cm of 3-3] (3-3b)  {\ {\small $\langle3,3\rangle$:\!\!}};
		\node[above = -0.1cm of 4-4] (4-4b)  {\ {\small $\langle4,4\rangle$:\!\!}};
		\node[above = -0.1cm of 5-6] (6-7b)  {\ {\small $\langle5,6\rangle$:\!\!}};
		\node[above = -0.1cm of 7-7] (7-7b)  {\ {\small $\langle7,7\rangle$:\!\!}};
		\node[above = -0.1cm of 8-8] (8-8b)  {\ {\small $\langle8,8\rangle$:\!\!}};
		\node[above = -0.1cm of 9-9] (9-9b)  {\ {\small $\langle9,9\rangle$:\!\!}};
		\node[above = -0.1cm of 5-5] (5-5b)  {\ {\small $\langle5,5\rangle$:\!\!}};
		\node[above = -0.1cm of 6-6] (6-6b)  {\ {\small $\langle6,6\rangle$:\!\!}};
		\node[below = -0.1cm of 1-1] (1-1b)  {$a$};
		\node[below = -0.1cm of 2-2] (2-2b)  {$b$};
		\node[below = -0.1cm of 3-3] (3-3b)  {$a$};
		\node[below = -0.1cm of 4-4] (4-4b)  {$b$};
		\node[below = -0.1cm of 7-7] (7-7b)  {{$c$}};
		\node[below = -0.1cm of 8-8] (8-8b)  {{$a$}};
		\node[below = -0.1cm of 9-9] (9-9b)  {{$b$}};
		\node[below = -0.1cm of 5-5] (5-5b)  {{$a$}};
		\node[below = -0.1cm of 6-6] (6-6b)  {{$b$}};

\draw (1-9) to (1-4);
\draw (1-9) to (5-9);
\draw (1-4) to (1-2);
\draw (1-4) to (3-4);
\draw (5-9) to (5-7);
\draw (5-9) to (8-9);
\draw (1-2) to (1-1);
\draw (1-2) to (2-2);
\draw (3-4) to (3-3);
\draw (3-4) to (4-4);
\draw (5-7) to (5-6);
\draw (5-7) to (7-7);
\draw (8-9) to (8-8);
\draw (8-9) to (9-9);
\draw (5-6) to (5-5);
\draw (5-6) to (6-6);
 
	\end{tikzpicture}	
		}

        	\caption{The trees $\bbT_0 = \bbT_{\{x_1,x_2\}}$ (on the left) and $\bbT_1 = \bbT_{\{x_2\}}$ built from $\cA^{\sf ex}$.}
	\label{fig:prep}
\end{figure}

This procedure is described in Algorithm~\ref{alg:prep}, albeit with a single subtlety. 
To prove item~(1) in Theorem~\ref{theo:words}, we do not need the full set $\{\bbT_Y\}_{Y\subseteq X}$ (although we do need it for item~(2)). Instead, we use just the sequence of $k+1$ trees $\bbT_0 := \bbT_{\{x_1,\dots,x_k\}}, \bbT_1 := \bbT_{\{x_2,\dots,x_k\}}, \ldots,$ $\bbT_{k-1} := \bbT_{\{x_k\}},\bbT_k := \bbT_\emptyset$.
This completes the description of the preprocessing~phase. 
\subparagraph{Logic of the stored matrices} We can interpret the logic of the data structure by an auxiliary set of mapping rules: for an $i\in [0,k]$ and $1 \leq l \leq r \leq n$, we define a set of mapping rules inside $[l,r]$ by
$
\tau^{[l,r]}_{> i} := \{x_1\Mapsto \emptyset\ ,\  \ldots\ ,\  x_{i}\Mapsto \emptyset\ ,\  x_{i+1}\Mapsto [l,r] \ ,\   \ldots\ ,\  x_k\Mapsto [l,r] \}.
$
Intuitively, this set is a simple filter that disallows any partial mapping that mentions any variable $x \prec x_i$.
For a simpler notation, instead of $\cM\langle l, r:\tau^{[l,r]}_{> i}\rangle$, we write $\cM\langle l, r:\tau_{> i}\rangle$.
\begin{lemma}\label{lem:matrices}
At the end of Algorithm~\ref{alg:prep}, the matrix stored in $\bbT_i\langle l,r\rangle$ is equal to $\cM\langle l, r:\tau_{> i}\rangle$ for every valid index $(l,r)$.
\end{lemma}
\subparagraph{Complexity analysis}
We briefly observe that every tree has a total of $2\cdot n - 1$ nodes, each of which holds a matrix in $\bbN^{Q\times Q}$ which is computed by a single matrix multiplication. As we store a total of $k+1$ trees, we conclude that the preprocessing phase takes $O(|Q|^{\omega}\cdot|X|\cdot|w|)$.

\subsection{The {\sc Update} subroutine}

Before moving on to describe the access phase in our algorithm, we introduce an auxiliary subroutine, described in Algorithm~\ref{alg:upd}. Lemma~\ref{lem:matrices} implies that, at the end of the preprocessing phase, the matrices in $\bbT_{i}$ will count partial runs that do not see any variable among $x_1,\ldots,x_i$. The goal of the {\sc Update} subroutine is to modify $\bbT_{i}$ so that these matrices count partial runs that see $x_1$ in position $s_1$, see $x_2$ in position $s_2$, and so on. We describe it as a procedure that modifies some version of a $\bbT_i$ (after possible updates) to include a single mapping $x\mapsto s$. In the access phase, in the $i$-th iteration (i.e., after computing the value $s_i$), we will call it once for each $\bbT_j$ such that $j > i$.

\begin{lemma}\label{lem:upd}
Fix $\bbT$ to be some $\bbT_i$ that has been possibly already updated. 
	Let $(l,r)$ be a valid index and define $\tau^{\sf curr}_{l,r}$ as the functional set of rules inside $[l,r]$ such that the matrix in $\bbT\langle l, r\rangle$ is equal to $\cM\langle l, r : \tau^{\sf curr}_{l,r}\rangle$; and assume that $(x\Mapsto\emptyset) \in \tau^{\sf curr}_{l,r}$ for some $x\in X$. Call $\bbT'$ the tree that results from performing {\sc Update}$(\bbT, x\mapsto s)$. Then:
	\begin{itemize}
	\item if $s\in[l,r]$, the matrix in $\bbT'\langle l, r\rangle$ is equal to $\cM\langle l, r : \tau^{\sf new}_{l,r}\rangle$ where $\tau^{\sf new}_{l,r}$ is the set that results from replacing $x\Mapsto \emptyset$ in $\tau^{\sf curr}_{l,r}$ by $x\mapsto s$.
	\item if $s\not\in[l,r]$, the matrix in $\bbT'\langle l, r\rangle$ is equal to $\cM\langle l, r : \tau^{\sf curr}_{l,r}\rangle$.
	\end{itemize}
\end{lemma}

\begin{algorithm}[t]
	\caption{Updating a tree $\bbT$ to include a mapping $x\mapsto s$.}\label{alg:upd}	
	\smallskip
	\begin{algorithmic}[1]
		\Procedure{{\sc Update}}{$\bbT, x\mapsto s, l, r $} 
		\If{$s < l$ {\bf or} $r < s$}
		\State {\bf let} $\bbT_i\langle l,r\rangle::M$
		\State {\bf return} $M$
		\ElsIf{$l = r$}
		\State $a\gets w[l]$
		\State $Z \gets \textsf{vars at }\bbT_i\langle l,r\rangle$
		\State $\bbT_i\langle l,r\rangle ::M \gets \Delta^{{Z \cup \{x\}}}_a(\cA)$
		\State $\textsf{vars at }\bbT_i\langle l,r\rangle \gets Z\cup \{x\}$
		\State {\bf return} $M$
		\Else 
		\State $m \gets {\sf mid}(l, r)$
		\State $M_1\gets \textsc{Update}(\bbT, x\mapsto s, l, m)$
		\State $M_2\gets \textsc{Update}(\bbT, x\mapsto s, m+1, r)$
		\State $\bbT_i\langle l, r\rangle  :: M \gets M_1 \cdot M_2$
		\State {\bf return} $M$
		\EndIf
		\EndProcedure
	\end{algorithmic}
\end{algorithm}

\subsection{Access phase}

We now describe the access phase that is actually used in the algorithm, by basing ourselves on the template of the access phase (Section~\ref{sec:template}). We will give a summarized explanation of the algorithm that doubles as a correctness proof for item (1) in Theorem~\ref{theo:words}.
\subparagraph{Explanation.} 
Let $\mu^*$ be the $t$-th mapping in $\sem{\cA}(w)$ and fix $s_1 := \mu^*(x_1),\ldots, s_k := \mu^*(x_k)$.
For $i\in[1,k]$ and $m\in[1,n]$, define $\tau_{x_i \leq m} :=\{x_1\mapsto s_1\, ,\, \ldots,x_{i-1}\mapsto s_{i-1}\, ,\, x_i\Mapsto[1,m]\}$.
Note how, in the template of the access phase, the required counting operations can be summarized as computing $\sharp\sem{\cA}(w)\langle\tau_{x_i \leq m} \rangle$ for some $i$ and $m$.

Assume that we are at iteration $i$, and we have computed the values $s_1,\ldots,s_{i-1}$, and also performed the updates on lines~\ref{alg:upd1} and \ref{alg:upd2} of Algorithm~\ref{alg:acc2}.
We note that, for any $m$, we can identify a $\tau_L$ inside $[1,m]$ and a $\tau_R$ inside $[m+1,n]$ such that $\tau_L\cdot\tau_R = \tau_{x_i \leq m}$; then, we define $M_L^* := \cM\langle 1, m : \tau_R\rangle$ and $M_R^* := \cM\langle m+1, r : \tau_L\rangle$ and we observe that $\sharp\sem{\cA}(w)\langle\tau_{x_i \leq m} \rangle = \sum_{p\in I , q \in F}(M_L^*\cdot M_R^*)$ by Lemma~\ref{lem:ms}.

Let us look at $\tau_L$ and $\tau_R$ more carefully. We see that $\tau_L = \{x_j\mapsto s_j\}_{s_j \in [1,m]}\cup\{x_j\Mapsto \emptyset\}_{s_j \in [m+1,n]}\cup\{x_i\Mapsto[1,m]\}$; and that $\tau_L = \{x_j\mapsto s_j\}_{s_j \in [m+1,n]}\cup\{x_j\Mapsto \emptyset\}_{s_j \in [1,m]}\cup\{x_i\Mapsto\emptyset\}$. By Lemma~\ref{lem:upd}, the updates that were performed ensure that we can find this $M_L$ by composing matrices from the current version of $\bbT_{i-1}$, and we can find  $M_R$ by composing matrices from the current version of $\bbT_{i}$. More precisely, $M_L^*$ can be obtained by composing the matrices stored at $\bbT_{i-1}\langle 1,m_1\rangle, \bbT_{i-1}\langle m_1+1,m_2\rangle,\ldots,\bbT_{i-1}\langle m_{\chi-1}+1,m_{\chi}\rangle, \bbT_{i-1}\langle m_{\chi}+1,m\rangle$ for any choice of $1 \leq m_1 <  \cdots  < m_{\chi}' < n$; and $M_R^*$ can be obtained by composing the matrices stored at $\bbT_{i}\langle m,m_1'\rangle, \bbT_{i}\langle m_1'+1,m_2'\rangle,\ldots,\bbT_{i}\langle m_{\chi-1}'+1,m_{\chi}'\rangle, \bbT_{i}\langle m_{\chi}'+1,n\rangle$ for any choice of $m < m_1' < \cdots < m_{\chi}' < n$.

The last point to explain is how to find these matrices effciently.\footnote{We note that, without doing anything further, the explanation so far does give us a logsquare access algorithm: since {\sf mid} is strongly balanced, we can compose any range by logarithmically many of the ranges realized by {\sf mid} and use these to compute $M_L^*$ and $M_R^*$.} 
We observe that $Algorithm~\ref{alg:acc2}$ defines two matrices $M_L^{\sf out}$ and $M_R^{\sf out}$ that work as caches of a prefix in $[1,m]$, and a suffix in $[m+1, n]$ respectively. With these, we only need to compose a constant number of matrices  (line~\ref{alg:comp}) to obtain $M_L^*$ and $M_R^*$ at each step in the binary search. The value that is removed from $t$ at the end of each iteration is also obtained by using the cached matrices (line~\ref{alg:diff}).

\begin{algorithm}[t]
	\caption{Access the $t$-th element in $\sem{\cA}(w)$.}\label{alg:acc2}	
	\smallskip
	\begin{algorithmic}[1]
		\hspace{-2.5em}
		\begin{varwidth}[t]{0.45\textwidth}
		\Procedure{{\sc Access}}{$\cA, \bbT_1, \ldots, \bbT_k, t$}
		\State {\bf declare} $s_1,\ldots,s_k$
		\For{$i \in [1,k]$}
		\State $(s_i, M^{\sf out}_L, M^{\sf out}_R)\gets\textsc{BinSearch}$
		\State {\bf let} $\bbT_{i}\langle s_i, s_i\rangle :: M^{\sf in}$
		\State $M \gets M^{\sf out}_L \cdot M^{\sf in} \cdot M^{\sf out}_R$
		\State ${\sf diff} \gets \sum_{p\in I, q\in F}M(p,q)$~\label{alg:diff}
		\State $t \gets t - {\sf diff} $
		\For{$j \in [i+1, n]$}\label{alg:upd1}
		\State $\textsc{Update}(\bbT_{j}, x_i\mapsto s_i, 1, n)$\label{alg:upd2}
		\EndFor
		\EndFor
		\State {\bf return} $\{x_1\mapsto s_1\, ,\, \ldots \, , \, x_k \mapsto s_k\}$
		\EndProcedure
		\end{varwidth}
		\hspace{2.5em}
		\begin{varwidth}[t]{0.55\textwidth}
		\Procedure{{\sc BinSearch}}{}
		\State $M^{\sf out}_L \gets I_Q\ ; \ M^{\sf out}_R \gets I_Q$
		\State $(l,r) \gets (1, n)$
		\While{$l < r$}
		\State $m \gets {\sf mid}(l, r)$
		\State {\bf let} $\bbT_{i-1}\langle l, m\rangle :: M^{\sf in}_L$ 
		\State {\bf let} $\bbT_{i}\langle m+1, r\rangle :: M^{\sf in}_R$
		\State $M \gets M^{\sf out}_L \cdot M^{\sf in}_L \cdot M^{\sf in}_R \cdot M^{\sf out}_R$\label{alg:comp}
		\If{$t \leq \sum_{p\in I, q\in F}M(p,q)$}
		\State $(l, r) \gets (l, m)$
		\State $M^{\sf out}_R \gets M^{\sf in}_R \cdot M^{\sf out}_R$
		\Else
		\State $(l, r) \gets (m+1, r)$
		\State $M^{\sf out}_L \gets M^{\sf out}_L \cdot M^{\sf in}_L$
		\EndIf
		\EndWhile
		\State {\bf return} $(l, M^{\sf out}_L, M^{\sf out}_R)$
		\EndProcedure
		\end{varwidth}
	\end{algorithmic}
\end{algorithm}

\subparagraph{Restoring the data structure}
To perform the access operation more than once, the original data structure must be preserved.
We observe that this can be maintained by storing the indices that were modified, which amount to at most $(k+1)\cdot c_{\sf mid}\cdot\log(n)$ matrices, and then restoring them. This concludes the description of the main algorithm.

\subparagraph{Complexity analysis}
The runtime of the access phase is dominated by the \textsc{Update} operation, which takes $O(|Q|^{\omega}\cdot|X|^2\cdot\log(|w|))$ time.

\subsection{If the order $\prec$ is given in the access phase}

We describe how to address case (2) in Theorem~\ref{theo:words}. The main difference is that instead of only the trees $\bbT_0,\ldots,\bbT_k$, we keep the full collection $\{\bbT_Y\}_{Y\subseteq X}$. Then, in the access phase, we identify the variables in the sequence $y_1 \prec \cdots \prec y_k$ and obtain the mapping $\mu^*$ using trees $\bbT_{\{y_1,\ldots,y_k\}}, \bbT_{\{y_2,\ldots,y_k\}}, \ldots, \bbT_{\{y_k\}}$. After identifying these trees, the procedure is fully analogous to case (1). This completes the proof of Theorem~\ref{theo:words}.

\section{Extension to SLPs}\label{sec:slps}

In this section, we extend the algorithm to receive a string that has been compressed by an SLP-compression scheme. The result is as follows:

\begin{theorem}\label{theo:slps}
	Let $\cA$ be an unambiguous vset automata on a set of variables $X$ and let $\frakS$ be an SLP that represents a string $w\in\Sigma^*$.
	\begin{enumerate}
	\item For a fixed order $\prec$ over $X$, we can have direct (ranked) access on the set $\sem{\cA}(w)$  in $O(|Q|^{\omega}\cdot|X|^2\cdot\log(|w|))$ time
	after $O(|Q|^{\omega}\cdot|X|\cdot|\frakS|)$ preprocessing.
	\item If $\prec$ is given as input in the access phase, the problem requires $O(|Q|^{\omega}\cdot 2^{|X|}\cdot|\frakS|)$ 
	preprocessing time. 
	\end{enumerate}
\end{theorem}

The rest of the section is dedicated to proving this theorem. As such, we fix $\cA$ and $\frakS$, where $\str(\frakS) = w$ and $|w| = n$. 
We forewarn that the algorithm will be largely analogous to the string case.

We start by highlighting a crucial balancing result for SLPs. By virtue of it, we assume that the input SLP $\frakS$ has depth $O(\log n)$ for the rest of the section.

\begin{theorem}[\cite{slpbalancing}]
Given an SLP $\frakS$ that represents a string of length $n$, one can construct in linear time an equivalent SLP of size $O(|\frakS|)$ and depth $O(\log n)$.
\end{theorem}

We point out that this result does not imply that the resulting SLP is strongly balanced. The requirement of a function {\sf mid} that is strongly balanced is not necessary for the results in Section~\ref{sec:words}; it will only be important in the next section, where we deal with edits in the data.

Next, we observe that we can precompute the values of $|\str(A)|$ for $A \in N$ by a linear-time dynamic-programming approach on the DAG structure of the SLP. This idea is standard in SLP literature~\cite{slpedits, slpenum1, slpenum2}.

\subsection{Preprocessing}

The preprocessing phase is presented in Algorithm~\ref{alg:slpprep}. We switch from trees $\bbT_{Y}$ to graph structures $\bbD_{Y}$ whose structure is dictated by the SLP $\frakS$ itself. 
The indexing is no longer given by ranges but rather by nonterminals from $\frakS$. 
The information stored at each index is still a matrix in $\bbN^{Q\times Q}$ and a set of variables in $X$; yet the indices themselves change from ranges to nonterminals in $N$.
The idea of the preprocessing is mostly unchanged,
and the way they work can be directly borrowed from the strings case: every $\bbD_{Y}$ can be ``unrolled'' into a tree $\bbT_{Y}$, so the values we obtain from $\bbD_{Y}$ satisfy the necessary conditions for correctness.

\begin{algorithm}[t]
	\caption{Preprocessing phase for direct access on $\sem{\cA}(\str(\frakS))$.}\label{alg:slpprep}	
	\smallskip
	\begin{algorithmic}[1]
		\hspace{-2.5em}
		\begin{varwidth}[t]{0.5\textwidth}
		\Procedure{{\sc Preprocessing}}{$\cA, \frakS$} 
		\State {\bf declare} $\bbD_1,\ldots,\bbD_k$
		\For{$i \in [0,k]$}
		\State $\textsc{Build}(\bbD_i, S_0)$
		\EndFor
		\EndProcedure
		\end{varwidth}
		\hspace{2.5em}
		\begin{varwidth}[t]{0.5\textwidth}
		\Procedure{{\sc Build}}{$\bbD_i, A$} 
		\If{$R(A) = a$}
		\State $\bbD_i\langle A\rangle::M \gets \Delta^{\{x_{i+1}\, ,\,\ldots\, ,\,x_k\}}_a$
		\State {\bf return} $M$
		\Else 
		\State $BC\gets R(A)$
		\State $M_1\gets \textsc{Build}(\bbT_i, B)$
		\State $M_2\gets \textsc{Build}(\bbT_i, C)$
		\State $\bbD_i\langle A\rangle :: M \gets M_1 \cdot M_2$
		\State {\bf return} $M$
		\EndIf
		\EndProcedure
		\end{varwidth}
	\end{algorithmic}
\end{algorithm}

\subsection{Access phase}

The access phase is described in Algorithm~\ref{alg:accslps}.
Again, the idea is largely unchanged, and it is straightforward to see why the first iteration retrieves the correct value for $s_1$ with an argument analogous to the case for strings.

However, the \textsc{Update} operation operates differently: the matrices are no longer associated to a range, hence modifying a matrix at a leaf and propagating this value upwards could affect other indices indiscriminately.
For our \textsc{Update} operation, we modify the structures $\bbD_{i},\ldots,\bbD_{k}$ borrowing ideas from the existing literature on SLP updates~\cite{slpedits,compressedtreesenum}.

\newcommand{\waeuat}{0.65}

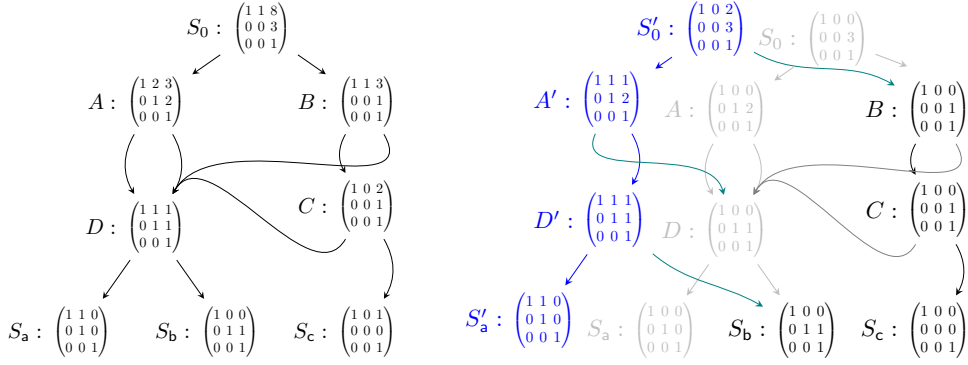
\begin{figure}[t]
	\centering
	\resizebox{5.9cm}{!}{
				\begin{tikzpicture}[scale=1.7,->,>=stealth,]

		\setlength\arraycolsep{2pt}
		\renewcommand\arraystretch{0.9}

		\node (nS0) at (0,2.7) {
\scalebox{\waeuat}{
$\begin{pmatrix}
	1 & 1 & 8\\
	0 & 0 & 3\\
	0 & 0 & 1\\
\end{pmatrix}$}};
		\node (nA) at (-1,2) {
\scalebox{\waeuat}{
$\begin{pmatrix}
	1 & 2 & 3\\
	0 & 1 & 2\\
	0 & 0 & 1\\
\end{pmatrix}$}};
		\node (nB) at (1,2) {
\scalebox{\waeuat}{
$\begin{pmatrix}
	1 & 1 & 3\\
	0 & 0 & 1\\
	0 & 0 & 1\\
\end{pmatrix}$}};
		\node (nC) at (1,1) {
\scalebox{\waeuat}{
$\begin{pmatrix}
	1 & 0 & 2\\
	0 & 0 & 1\\
	0 & 0 & 1\\
\end{pmatrix}$}};
		\node (nD) at (-1,0.8) {
\scalebox{\waeuat}{
$\begin{pmatrix}
	1 & 1 & 1\\
	0 & 1 & 1\\
	0 & 0 & 1\\
\end{pmatrix}$}};
		\node (Sa) at (-1.7,-0.2) {
\scalebox{\waeuat}{
$\begin{pmatrix}
	1 & 1 & 0\\
	0 & 1 & 0\\
	0 & 0 & 1\\
\end{pmatrix}$}};
		\node (Sb) at (-0.3,-0.2) {
\scalebox{\waeuat}{
$\begin{pmatrix}
	1 & 0 & 0\\
	0 & 1 & 1\\
	0 & 0 & 1\\
\end{pmatrix}$}};
		\node (Sc) at (1,-0.2) {
\scalebox{\waeuat}{
$\begin{pmatrix}
	1 & 0 & 1\\
	0 & 0 & 0\\
	0 & 0 & 1\\
\end{pmatrix}$}};

\node[left = -0.15cm of Sc] (Scl) {$S_{\sf c}:$};
\node[left = -0.15cm of Sb] (Sbl) {$S_{\sf b}:$};
\node[left = -0.15cm of Sa] (Sal) {$S_{\sf a}:$};
\node[left = -0.15cm of nD] (nDl) {$D:$};
\node[left = -0.15cm of nC] (nCl) {$C:$};
\node[left = -0.15cm of nB] (nBl) {$B:$};
\node[left = -0.15cm of nA] (nAl) {$A:$};
\node[left = -0.15cm of nS0] (nS0l) {$S_0:$};
		
		\draw (nS0) to (nA);
		\draw (nS0) to (nB);
		\draw (nA) to[bend right] (nD);
		\draw (nA) to[bend left] (nD);
		\draw (nB) to[bend right] (nC);
		\draw (nB) to[out=-60, in=60] (nD);
		\draw (nC) to[out=-120, in=60] (nD);
		\draw (nC) to[bend left] (Sc);
		\draw (nD) to (Sa);
		\draw (nD) to (Sb);
	\end{tikzpicture}
		}
	\resizebox{7.4cm}{!}{
				\begin{tikzpicture}[scale=1.5,->,>=stealth,inner sep=0.1cm]

		\setlength\arraycolsep{2pt}
		\renewcommand\arraystretch{0.9}

		\node[rectangle] (nS0) at (0, 2.7) {
		
\scalebox{\waeuat}{
$\color{lightgray}{\begin{pmatrix}
	1 & 0 & 0\\
	0 & 0 & 3\\
	0 & 0 & 1\\
\end{pmatrix}}$}};
		\draw[color=lightgray] (nS0) to (nB);
		
		\node[rectangle] (nS0p) at (-1.2, 2.8) {
		
\scalebox{\waeuat}{
$\color{blue}{\begin{pmatrix}
	1 & 0 & 2\\
	0 & 0 & 3\\
	0 & 0 & 1\\
\end{pmatrix}}$}};
		\node[rectangle] (nA) at (-1, 2) {
		
\scalebox{\waeuat}{
$\color{lightgray}{\begin{pmatrix}
	1 & 0 & 0\\
	0 & 1 & 2\\
	0 & 0 & 1\\
\end{pmatrix}}$}};
		\node[rectangle] (nAp) at (-2.2, 2.1) {
	
\scalebox{\waeuat}{	
$\color{blue}{
\begin{pmatrix}
	1 & 1 & 1\\
	0 & 1 & 2\\
	0 & 0 & 1\\
\end{pmatrix}
}$
}
};
		\node[rectangle] (nB) at (1,2) {
\scalebox{\waeuat}{
$\begin{pmatrix}
	1 & 0 & 0\\
	0 & 0 & 1\\
	0 & 0 & 1\\
\end{pmatrix}$}};
		\node[rectangle] (nC) at (1,1) {
\scalebox{\waeuat}{
$\begin{pmatrix}
	1 & 0 & 0\\
	0 & 0 & 1\\
	0 & 0 & 1\\
\end{pmatrix}$}};
		\node[rectangle] (nD) at (-1,0.8) {

\scalebox{\waeuat}{		
$\color{lightgray}{\begin{pmatrix}
	1 & 0 & 0\\
	0 & 1 & 1\\
	0 & 0 & 1\\
\end{pmatrix}}$}};
		\node[rectangle] (nDp) at (-2.2, 0.9) {

\scalebox{\waeuat}{		
$\color{blue}{\begin{pmatrix}
	1 & 1 & 1\\
	0 & 1 & 1\\
	0 & 0 & 1\\
\end{pmatrix}}$}};
		\node[rectangle] (Sa) at (-1.7,-0.2) {

\scalebox{\waeuat}{		
$\color{lightgray}{\begin{pmatrix}
	1 & 0 & 0\\
	0 & 1 & 0\\
	0 & 0 & 1\\
\end{pmatrix}}$}};
		\node[rectangle] (Sap) at (-2.9,-0.1) {

\scalebox{\waeuat}{		
$\color{blue}{\begin{pmatrix}
	1 & 1 & 0\\
	0 & 1 & 0\\
	0 & 0 & 1\\
\end{pmatrix}}$}};
		\node[rectangle] (Sb) at (-0.3,-0.2) {
\scalebox{\waeuat}{
$\begin{pmatrix}
	1 & 0 & 0\\
	0 & 1 & 1\\
	0 & 0 & 1\\
\end{pmatrix}$}};
		\node[rectangle] (Sc) at (1,-0.2) {
\scalebox{\waeuat}{
$\begin{pmatrix}
	1 & 0 & 0\\
	0 & 0 & 0\\
	0 & 0 & 1\\
\end{pmatrix}$}};

\node[left = -0.15cm of Sc] (Scl) {$S_{\sf c}:$};
\node[left = -0.15cm of Sb] (Sbl) {$S_{\sf b}:$};
\node[left = -0.15cm of Sa] (Sal) {\color{lightgray}{$S_{\sf a}:$}};
\node[left = -0.15cm of Sap] (Sapl) {\color{blue}{$S'_{\sf a}:$}};
\node[left = -0.15cm of nD] (nDl) {\color{lightgray}{$D:$}};
\node[left = -0.15cm of nDp] (nDpl) {\color{blue}{$D':$}};
\node[left = -0.15cm of nC] (nCl) {$C:$};
\node[left = -0.15cm of nB] (nBl) {$B:$};
\node[left = -0.15cm of nA] (nAl) {\color{lightgray}{$A:$}};
\node[left = -0.15cm of nAp] (nApl) {\color{blue}{$A':$}};
\node[left = -0.15cm of nS0] (nS0l) {\color{lightgray}{$S_0:$}};
\node[left = -0.15cm of nS0p] (nS0pl) {\color{blue}{$S'_0:$}};
		
		\draw[color=lightgray] (nS0) to (nA);
		\draw[color=blue] (nS0p) to (nAp);
		\draw[color=teal] (nS0p) to[out=-30, in=150] (nB);
		\draw[color=lightgray]  (nA) to[bend right] (nD);
		\draw[color=lightgray]  (nA) to[bend left] (nD);
		\draw[color=blue]  (nAp) to[bend left] (nDp);
		\draw[color=teal]  (nAp) to[out=-115, in=105] (nD);
		\draw (nB) to[bend right] (nC);
		\draw[color=gray] (nB) to[out=-60, in=60] (nD);
		\draw[color=gray] (nC) to[out=-120, in=60] (nD);
		\draw (nC) to[bend left] (Sc);
		\draw[color=lightgray] (nD) to (Sa);
		\draw[color=blue] (nDp) to (Sap);
		\draw[color=lightgray]  (nD) to (Sb);
		\draw[color=teal] (nDp) to[out=-40, in=150] (Sb);
	\end{tikzpicture}
		}
	\caption{(left) The matrices as set in $\bbD_0$\label{fig:slp1}; (right) in gray and black, the matrices as set in $\bbD_1$, and in blue, the new matrices after performing $\textsc{Update}(\bbD_1, x\mapsto 3)$.}
	\label{fig-re}\label{fig:slp2}
\end{figure}

We will sketch the idea and provide an illustrated example. For \textsc{Update}$(\bbD, x\mapsto s)$, we identify the path of nonterminals $S_0,A_1,\ldots,A_m$ that leads to the $s$-th element in $w$. Then, we duplicate these nonterminals and connect them to the rest of the DAG so that the path to access the $s$-th element is this sequence of fresh nonterminals. In particular, $S_0$ is swapped by its duplicate. 
Then, we perform the update on $\bbD\langle A_m \rangle$ as we did in the string case, and propagate the information upwards. 
The resulting $\bbD$ contains the proper information to compute the values $s_2,\ldots,s_k$. 
The illustrated example is presented in Figure~\ref{fig:slp2}.

We conclude the section by noting that
the idea for item (2) in Theorem~\ref{theo:slps} is fully analogous to the case for strings.
Namely, when the order $\prec$ is given in the access phase we can use structures $\{\bbD_Y\}_{Y\subseteq X}$ instead.

\subparagraph{Complexity analysis}
For the preprocessing phase, it can be seen that Algorithm~\ref{alg:slpprep} takes $O(|Q|^{\omega}\cdot {|X|}^2\cdot|\frakS|)$ time.
For the access phase, we note that the binary search takes as many steps as the depth of $\frakS$, 
which is guaranteed to be in $O(\log n)$ due to balancedness; thus the access phase takes $O(|Q|^{\omega}\cdot|X|^2\cdot\log(|w|))$ time. This concludes the proof of Theorem~\ref{theo:slps}.

\begin{algorithm}[t]
	\caption{Access the $t$-th element in $\sem{\cA}(\str(\frakS))$.}\label{alg:accslps}	
	\smallskip
	\begin{algorithmic}[1]
	\hspace{-2.5em}
	\begin{varwidth}[t]{0.5\textwidth}
		\Procedure{{\sc Access}}{$\cA, \bbD_1, \ldots, \bbD_k, t$}
		\State $s_1,\ldots,s_k$
		\For{$i \in [1,k]$}
		\State $(s_i, a, M^{\sf out}_L, M^{\sf out}_R)\gets\textsc{BinSearch}$
		\State {\bf let} $\bbD_{i}\langle A\rangle :: M^{\sf in}$
		\State $M \gets M^{\sf out}_L \cdot M^{\sf in} \cdot M^{\sf out}_R$
		\State ${\sf diff} \gets \sum_{p\in I, q\in F}M(p,q)$
		\State $t \gets t - {\sf diff}$
		\For{$j \in [i, n]$}
		\State $\textsc{Update}(\bbD_{j}, x_i\mapsto s_i, S_0)$
		\EndFor
		\EndFor
		\State {\bf return} $\{x_1\mapsto s_1\, ,\, \ldots \, , \, x_k \mapsto s_k\}$
		\EndProcedure
	\end{varwidth}	
	\hspace{2.5em}
	\begin{varwidth}[t]{0.5\textwidth}
		\Procedure{{\sc BinSearch}}{}
		\State $M^{\sf out}_L \gets I_Q\ ; \ M^{\sf out}_R \gets I_Q$
		\State $(l,r) \gets (1, n)\ ; \ A\gets S_0$
		\While{$l < r$}
		\State $BC \gets R(A)$
		\State $m \gets l + |B| - 1$
		\State {\bf let} $\bbD_{i-1}\langle B\rangle :: M^{\sf in}_L$ 
		\State {\bf let} $\bbD_{i}\langle C\rangle :: M^{\sf in}_R$
		\State $M \gets M^{\sf out}_L \cdot M^{\sf in}_L \cdot M^{\sf in}_R \cdot M^{\sf out}_R$
		\If{$t \leq \sum_{p\in I, q\in F}M(p,q)$}
		\State $(l, r) \gets (l, m)$
		\State $M^{\sf out}_R \gets M^{\sf in}_R \cdot M^{\sf out}_R$
		\Else
		\State $(l, r) \gets (m+1, r)$
		\State $M^{\sf out}_L \gets M^{\sf out}_L \cdot M^{\sf in}_L$
		\EndIf
		\EndWhile
		\State {\bf return} $(l, R(A), M^{\sf out}_L, M^{\sf out}_R)$
		\EndProcedure
	\end{varwidth}
	\end{algorithmic}
\end{algorithm}

\section{Handling updates}\label{sec:upd}

In this section, we explain how the results of complex editing operations presented in~\cite{diraccmso} still hold, 
while preserving the logtime access time.
Instead of the framework proposed in~\cite{diraccmso}, we will present this as an adaptation of the {\em complex string editing}\footnote{In~\cite{slpedits} the term used for {\em strings} is {\em documents}. We do not pay attention to this change any further than this footnote. } framework of~\cite{slpedits}, given that we can borrow the SLP edit operations with little change. 
Regrettably, we lack the {\sf split} and {\sf cut} operations that featured in~\cite{diraccmso}; we claim that they are also computable, and we plan to include them in the full version of the current paper.

We will proceed in a similar fashion as~\cite{slpenum2}, where the authors presented the editing framework, and then briefly explained how the machinery in~\cite{slpedits} was directly applicable to their data structure.

As a disclaimer, we note that the results in this section only hold for {\em strongly balanced} SLPs. In general, an SLP cannot be balanced with the strongness guarantee without possibly incurring in a $\log n$ factor blowup in its size~\cite{ganardicompression}.
However, we note that the results hold for strings with no further assumptions: we can represent a word by a strongly balanced SLP by using the function defined in Section~\ref{sec:words}.

\subparagraph{Complex string editing}
A string database over $\Sigma$ is a finite collection of strings over $\Sigma$ where each string has a name. Formally, a \emph{string database} is a function $D: \{d_1, \ldots, d_m\} \rightarrow \Sigma^*$ where $\{d_1, \ldots, d_m\}$ is a finite set of \emph{string names} and $D(d_i)$ is the string assigned to the name $d_i$. We also say that $\{d_1, \ldots, d_m\}$ is the \emph{schema} of the string database $D$. We define the size of $D$ as $|D| = \sum_{i=1}^m |D(d_i)|$.

Given a string database $D$ we can create new strings by a sequence of text-editing operations.
From~\cite{slpedits, slpenum2}, we borrow the notion of a \emph{Complex String Editing expression} (CSE-expression for short) over string names $\{d_1, \ldots, d_m\}$, defined by the following syntax 
\newcommand{\asdhush}{-0.7pt}
\[
	\psi  :=   d_\ell, \ell\in[1, m]   \hspace{\asdhush}\mid\hspace{\asdhush}   {\sf concat}(\psi, \psi)    \hspace{\asdhush}\mid\hspace{\asdhush}    {\sf extrac}(\psi, l, r)  \hspace{\asdhush}\mid\hspace{\asdhush} 
	  {\sf delete}(\psi, l, r)    \hspace{\asdhush}\mid\hspace{\asdhush}  {\sf insertop}(\psi, \psi, k)    \hspace{\asdhush}\mid\hspace{\asdhush}   {\sf copyop}(\psi, l, r, k)
\]
where the values $l,r,k$ are {\em valid}. %
The number of string names and operations in $\psi$ is noted $|\psi|$.
The semantics of an CSE-expression $\psi$ is given by a function $\sem{\psi}$ from string databases over $\{d_1, \ldots, d_m\}$ to strings in $\Sigma^*$, and works recursively as follows. 
\[
\renewcommand{\arraystretch}{1.4}
\begin{array}{rcl}
	\sem{d_\ell}(D) & = & D(d_\ell) \\
	\sem{{\sf concat}(\psi_1, \psi_2)}(D) & = & \sem{\psi_1}(D) \cdot \sem{\psi_2}(D)  \\
	\sem{{\sf extract}(\psi, l, r)}(D) & = & \sem{\psi}(D)[l,r] \\
	\sem{{\sf delete}(\psi, l, r)}(D) & = & \sem{\psi}(D)[1,l-1] \cdot \sem{\psi}(D)[r+1\ldots] \\
	\sem{{\sf insertop}(\psi_1, \psi_2, k)}(D) &=& \sem{\psi_1}(D)[1,k-1]\cdot \sem{\psi_2}(D) \cdot \sem{\psi_1}(D)[k\ldots] \\
	\sem{{\sf copyop}(\psi, l, r, k)}(D) &=& \sem{\psi}(D)[1,k-1]\cdot \sem{\psi}(D)[l,r] \cdot \sem{\psi}(D)[k\ldots] 
\end{array}
\]
where $\psi$, $\psi_1$, and $\psi_2$ are CSE-expressions, and $l, r, k$ are valid indices in their respective strings; and $u[k\ldots]$ denotes the suffix of $u$ that starts at position $k$. If $l$, $r$, or $k$ are not valid, then $\sem{\psi}(D)$ is undefined---in the algorithm, we assume they are always valid since verifying this takes linear time in $|D|$ and $|\psi|$.
We define the {\it maximum intermediate string size} $|\max_{\psi}(D)|$ induced by a CSE-expression $\psi$ on a string database $D$ as the maximum size of $\sem{\psi}(D)$ for any sub-expression $\psi$ of $\psi$.

The link between string databases and SLPs is given by a {\em rooting function} that maps schemas and nonterminals in an SLP. Formally, a {\em rooting function} from $\{d_1,\ldots,d_m\}$ to $\frakS$ is a function $\mathfrak{r}:\{d_1,\ldots,d_m\} \to N$. They define the string database $D_{\frakS, \mathfrak{r}}$ given by $D_{\frakS, \mathfrak{r}}(d) =  \str(\mathfrak{r}(d))$.

To state the main result in this section, we introduce the idea of an {\em editable direct access structure} for a string database $D$ over $\{d_1,\ldots,d_m\}$ and a vset automata $\cA$. 
We say that $\bbD$ is a {\em editable direct access structure} for $D$ and $\cA$ if it allows direct (ranked) access to the set $\sem{\cA}(\sem{\psi}(D))$ for any CSE-expression $\psi$ over $\{d_1,\ldots,d_m\}$. 
Futher, we say that $\bbD$ allows {\em $f(\psi, \cdot)$ editing} and {\em $g(\psi, \cdot)$ direct access} if there is an $f(\psi, \cdot)$-time procedure that receives $\psi$ and modifies $\bbD$ into an editable direct access structure $\bbD_{\psi}$ that allows direct access to $\sem{\cA}(\sem{\psi}(D))$ in $g(\psi, \cdot)$ time.

\begin{theorem}
	Let $\cA$ be an unambiguous vset automata with state set $Q$ and set of variables $X$, let $\frakS$ be a strongly balanced SLP and let $\mathfrak{r}$ be a rooting function to $\frakS$. We can produce, in $O(|Q|^{\omega}\cdot|X|\cdot|\frakS|)$ time, an editable direct access structure $\bbD$ for $D_{\frakS, \mathfrak{r}}$ and $\cA$ that allows $O(|\psi|\cdot|Q|^{\omega}\cdot|X|\cdot \log |\max_{\psi}(D)|)$ editing and $O(|Q|^{\omega}\cdot|X|^2\cdot |\max_{\psi}(D)|)$ direct access.
\end{theorem}

\subsection{Edits on strings}

We framed our results on edits solely for SLPs given that they can be trivially adapted for words. However, we claim that the editing procedure on the string case can be simpler, and, in particular, be done in a way that keeps the number of nodes in the data structure unchanged. 
The idea would be to use classic results for editing AVL trees, as it is done in~\cite{diraccmso}.
We defer further details to a future full version of this paper.

\section{Conclusions and future work}\label{sec:concl}

We have presented a dynamic direct (ranked) access algorithm for MSO queries on strings and SLP-compressed strings. We improve on~\cite{diraccmso} by shaving a log-factor in the access procedure, and showing an extension to SLPs. 
We note that the PhD theses of Bagan~\cite{phdbagan} and Kazana~\cite{phdkazana} both already had a version of direct access on MSO queries (even for trees), so the main contribution of our work is in adding the dynamic aspect, and the extension to SLPs.

Something of a blind spot of our work with respect to Bourhis et. al.'s~\cite{diraccmso} is that we do not say anything about query evaluation for semigroups---we actually did not see how our result gives further insights here. It would be interesting to have a characterization of which semigroups support log-time direct access and which do not by some lower bound.

For future work, we think that the main question is whether our techniques presented here can be extended to dynamic direct access for MSO queries on trees. We believe this to be possible, for example, via~\cite{niewerthtrees}. It would also prove technically challenging since this framework does not preserve contiguity, which is crucial in our solution.

Also, there is a clear shortcoming in the fact that the SLP needs to be {\em strongly balanced} to admit logtime updates.
In particular, this means that it does not work for an arbitrary SLP. We leave open the question of whether some weaker balancedness criterion for SLPs could achieve this.

\bibliography{ref}

\appendix

\section{Proofs of Section~\ref{sec:words}}

\subsection{Proof of Lemma~\ref{lem:ms}}

\begin{proof}
Let $\rho_L$ be an $(l,m)$-partial run over $w$ from $p$ to some $q'$ and let $\rho_R$ is an $(m+1,r)$-partial run over $w$ from $q'$ to $q$; further, let $\rho_L$ respect $\tau_L$ and let $\rho_R$ respect $\tau_R$. We concatenate the runs and obtain $\rho_L\cdot\rho_R$ which is an $(l,r)$-partial run over $w$ that respects $\tau_L\cdot\tau_R$. Likewise, let $\rho$ be an $(l,r)$-partial run over $w$ that respects $\tau_L\cdot\tau_R$. We can separate it at $m$ into two partial runs $\tau_L,\rho_R$: they both respect $\tau_L\cdot\tau_R$, yet $\tau_L$ does not involve any value in $[m+1,r]$, and $\tau_R$ does not involve any value in $[l,m]$, so in particular, $\rho_L$ respects $\tau_L$ and $\rho_R$ respects $\tau_R$.
We obtain that
$
 \cM\langle l, r:\tau_L 
\cdot \tau_R\rangle(p,q) = \sum_{q'\in Q}\cM\langle l, m:\tau_L\rangle(p,q')\cdot \cM\langle m+1, r:\tau_R\rangle(q',q)
$.
\end{proof}

\end{document}